%% file: main_arxiv.tex
\documentclass[11pt]{article}
\usepackage[utf8]{inputenc}
\input{Packages_and_commands.tex}
\usepackage{algorithm}
\usepackage{algpseudocode}
\usepackage{comment}

\newcommand{\vzeros}{\boldsymbol{0}}
\newcommand{\mZeros}{\boldsymbol{O}}

\title{Testing for peer effects in short panels\\ without identifying the network structure}
\author{Tom Boot \\ University of Groningen \\ \texttt{t.boot@rug.nl} \and
Art\={u}ras Juodis \\ University of Amsterdam and Tinbergen Institute \\ \texttt{a.juodis@uva.nl}}
\date{\today}

\begin{document}

\maketitle
\setlength{\abovedisplayskip}{3pt plus 1pt minus 1pt}
    \setlength{\belowdisplayskip}{3pt plus 1pt minus 1pt}
    \setlength{\abovedisplayshortskip}{3pt plus 1pt minus 1pt}
    \setlength{\belowdisplayshortskip}{3pt plus 1pt minus 1pt}

\begin{abstract}
We propose an instrument-based test for peer effects in short-\(T\) panels when the underlying network is unknown. The test avoids estimating the network by aggregating pairwise comparisons between covariate-projected outcomes of potential peers and residualized own outcomes. The statistic is designed for fixed-\(T\), large-\(N\) settings with individual fixed effects, heterogeneous slopes, heteroskedasticity, and serial dependence within units. We also develop a weighted extension that accounts for additive common time effects. The limiting distribution is standard Gaussian when the norm of the projected component of the outcomes is bounded away from zero, but can be nonstandard when it is local to zero. We derive a conditional Gaussian limit experiment and a feasible conditional Monte Carlo procedure to obtain uniformly valid critical values. Monte Carlo evidence shows that the conditional procedure leads to accurate size control while retaining power against latent peer-effect alternatives.
\end{abstract}
\newpage
\section{Introduction}
Outcomes of individuals, firms, schools, or regions may depend not only on their own characteristics, but also on the outcomes of others. When the network through which these interactions occur is observed, we can measure such peer effects, as in \citet{bramoulle2009identification} and the related discussion in \citet{goldsmithpinkham2013social}; see also the survey by \citet{bramoulle2020peer}. In many empirical settings, however, the relevant network is only partially observed, measured with error, endogenous, or entirely unknown. One option then is to estimate the network \citep{depaula2025identifying}, generally requiring a researcher to observe outcomes over a sufficiently long time period over which the network is stable. This raises the question: can one test for the presence of peer effects in short panels without identifying the network?

In this paper, we propose an instrument-based test for peer effects in fixed-\(T\) panels when the underlying network is unknown. We consider a linear model in which every unit is  potentially affected by every other unit. The null hypothesis is that all peer-effect coefficients are zero, and can be easily adapted to test a subset of the coefficients. The test avoids estimating the network by aggregating pairwise comparisons between projected outcomes of potential peers and residualized own outcomes. The resulting statistic can be interpreted both as an instrument-based test for peer effects and as a specific cross-sectional dependence statistic. The procedure is designed for fixed-\(T\), large-\(N\) panels with individual fixed effects, heteroskedasticity, and serial dependence within units.

The main theoretical complication that arises is that a Gaussian approximation is not uniformly valid over the nuisance parameter space. When the norm of the projected component is bounded away from zero, it is straightforward to see that the statistic admits a Gaussian limit. When the projected component is close to zero, however, the leading term may become degenerate and the limiting distribution can be nonstandard. We capture this behavior through a local Gaussian limit experiment. To attain uniformly valid inference, the key object is a conditioning statistic that removes the local nuisance parameter from the conditional distribution of the test statistic. This gives an exact conditional critical-value procedure in the limit experiment and a feasible conditional Monte Carlo implementation in finite samples.

The limit experiment can be mapped into the limit experiment arising from the two-stage least-squares score in an IV model under weak identification as in \citet{kleibergen2002pivotal}. This also opens up the possibility to construct a pivotal score statistic. We show that neither the score nor the conditional statistic dominates the other.\footnote{Despite this connection, we leave further exploration of the analogy with weak IV, including for example alternative tests as in \citet{anderson1949estimation} or \citet{moreira2003conditional}, for future work.}

Evidently, there are many sources of apparent peer effects, not all of them interesting. In particular, additive time effects can create the appearance of network connectedness. While individual effects are easily accommodated within our framework, time effects are well known to be challenging as removing them induces cross-sectional dependence in residualized outcomes \citep{juodis2022incidental,PesaranXie2026}. Leveraging on the information from observed regressors, we develop a deterministic analogue to the Rademacher weighting in \citet{juodis2022incidental}. The weighting ensures that after averaging over all individual units, the time effects are exactly removed from the model without mechanically inducing dependence by using a projection approach. The weights are not unique, and we provide guidance on how to set the weights to yield power against specific alternatives.

The fixed-\(T\) analysis we consider differs from the existing literature. The closest paper on peer effect testing is \citet{jung2026testing} who develop an Anderson-Rubin statistic based on instruments constructed from all potential peers' covariates. This is attractive because it avoids misspecifying the network, but it creates a many-instrument problem: as the number of potential links grows with the cross-sectional dimension, so does the number of instruments. This leads to a Bekker-style many-IV setting \citep{bekker1994alternative, mikusheva2022inference}. Importantly, the theory requires \(T\) to grow at least as fast as the dimension of the cross-section \(N\). Our construction uses the same identifying idea differently. Rather than forming one high-dimensional many-IV statistic, we aggregate many low-dimensional pairwise scores. This avoids a high-dimensional first-stage projection and yields a fixed-\(T\) asymptotic framework.

Our suggested testing approach is also related to tests for cross-sectional dependence in panels. Classical LM and CD-type tests aggregate residual covariances across units \citep{breusch1980lagrange,pesaran2021general,pesaran2015testing}. Our statistic has a similar aggregate form, but it is not a residual-residual covariance test. It correlates the residualized outcome of one unit with the projected outcome of another unit, where the projection is formed using covariates of the other unit. Unlike existing cross-sectional dependence tests, we consider a fixed-\(T\) framework, allow for autocorrelated residuals, and exploit a conditional limit rather than using a bootstrap as in \citet{hounyo2025serial}.

Finally, the paper relates to work on unobserved, partially observed, or misspecified networks. \citet{manresa2016estimating} estimates the structure of social interactions from panel data, while \citet{depaula2025identifying} provide conditions under which network ties and social interaction parameters can be recovered from panels without direct network observations. \citet{yauck2022estimation} studies peer-effect estimation with sampled networks and missing links. These papers ask what can be recovered about the network or about peer-effect parameters under additional structure. Our goal is more limited but also less demanding: we test the global null of no peer effects without recovering the network and without imposing sparsity, symmetry, row-normalization, or a parametric model for links.

Numerically, we first study the Gaussian limit experiment directly. We consider two stylized designs that differ in the correlation structure between the projected and residualized outcomes. The simulations show that Gaussian critical values can be misleading when the realized projected component is weak. We also illustrate the importance of normalization of the test statistic. While a straightforward self-normalization choice is available, this choice has the drawback that it can be close to zero or even negative. We therefore propose two alternatively normalizations that guarantee positivity. This positivity is a key requirement of the asymptotic arguments to go through. In the simulations, both alternative normalizations outperform the standard normalization that can be undefined in over half of the draws under a particular data generating process. The conditional procedure with positive normalization avoids these issues and delivers reliable size control across the designs.

Our paper also contains an illustrative Monte Carlo study. We find that the feasible plug-in procedure tracks the limit-experiment behavior well. Under the null, empirical rejection probabilities are close to the nominal level in the sample sizes considered. We show that the network structure plays a key role in how much power is lost, with substantial losses under a dense design, while only minor losses under a sparse design.  

As an empirical illustration, we revisit the study on the relative wage gap between high-skilled and low-skilled jobs in US manufacturing as considered in \citet{voigtlander2014skill,yin2021focused,Juodis2022CCER}. We use a rolling window of 8 years to document periods in which the test statistic indicates dependence after accounting for a small set of covariates, individual effects and individual-specific trends. After accounting for the latter, we observe that our weighted and unweighted statistics are close, suggesting only a minor role for time effects in this dataset. We document a prolonged period of dependence between US manufacturing sectors post-2000. Using a decomposition of the test-statistic, we find that this dependence is mainly driven by the relative wage gap in the food/tabacco sector affecting the relative wage gap in the machinery/electronics/instruments sector, while we do not find evidence for the reverse relation.

The rest of the paper is organized as follows. Section~\ref{sec:model_test} introduces the unknown-network peer-effects model and develops the proposed test statistic. Section~\ref{sec:theory} derives the fixed-\(T\) theory, presents the conditional inference procedure. Section~\ref{sec:monte_carlo} reports Monte Carlo evidence. In Section~\ref{sec:application_relative_wages} we illustrate our methodology using real data. Finally, Section~\ref{sec:conclusion} concludes.
\paragraph{Notation}
All vectors are column vectors unless stated otherwise. The vector \(\biota_T\) denotes a \(T\times1\) vector of ones, \(\be_i\) denotes a vector with element \(i\) equal to one and zeros elsewhere, and \(\bI\) denotes the identity matrix. The elementwise product between conformable matrices is denoted by \(\odot\). A single sum over all \(j\) not equal to \(i\) is written as \(\sum_{j\neq i}\). Finally, a.s. stands for statements that hold \emph{almost surely}, while $\perp$ is used to denote independence between two sets of random variables.

\section{Model and proposed test}
\label{sec:model_test}

\subsection{A linear peer-effects model}
\label{subsec:model}

For units \(i=1,\ldots,N\), we observe an outcome vector
$
\by_i=(y_{i1},\ldots,y_{iT})'\in\mathbb R^T$
and a matrix of time-varying regressors
$
\bX_i=(\bx_{i1},\ldots,\bx_{iT})'\in\mathbb R^{T\times p}.
$
The outcome vector is generated by the linear model
\begin{equation}
\label{eq:model}
\by_i
=
\sum_{j\neq i}\alpha_{ij}\by_j
+
\bX_i\bbeta_i
+
\lambda_i\biota_T
+
\bu_i,
\qquad i=1,\ldots,N,
\end{equation}
where \(\lambda_i\) is an individual fixed effect, \(\bbeta_i\) is a unit-specific regression coefficient vector, and \(\bu_i\in\mathbb R^T\) is a vector of idiosyncratic errors. We assume that \((\bu_i,\bX_i)\in\mathbb R^{T\times (p+1)}\) are independent across $i$, and that $\E[\bu_i\mid\bX_i]=\bzero_T$. The baseline model includes individual fixed effects but not time fixed effects. We treat a model with time effects as an extension to this baseline in Section \ref{subsec:time_effects_weighted}.

The coefficients \(\alpha_{ij}\) summarize the effect of unit \(j\)'s outcome on unit \(i\)'s outcome. We do not assume that the researcher observes which of these coefficients are zero, nor do we impose sparsity, symmetry, row-normalization, or a parametric model for the network. Thus the formulation in \eqref{eq:model} should be viewed as an unrestricted representation of possible peer effects.

The null hypothesis is
\begin{equation}
\label{eq:null}
\mathbb{H}_0:\alpha_{ij}=0
\quad\text{for all } i=1,\ldots,N,\quad j\neq i.
\end{equation}
Under the null, the original model collapses to
\begin{equation}
\label{eq:null_model}
\by_i=\bX_i\bbeta_i+\lambda_i\biota_T+\bu_i.
\end{equation}
The purpose of the test is to detect whether the data contain systematic evidence against \eqref{eq:null}. However, the procedure is not designed to estimate the network or identify which links are active. It is important to note that without additional restrictions on the network or on the form of social interactions, the null of no peer effects can be tested, but the endogenous and exogenous components of social interactions cannot generally be separated.

\subsection{The test statistic}

The statistic is built from two unit-specific transformations. First define the projection matrix onto the span of the observed time-varying regressors,
\begin{equation}
\label{eq:Pi}
\bP_i
:=
\bX_i(\bX_i'\bX_i)^{-1}\bX_i',
\end{equation}
and define the augmented regressor matrix
\begin{equation}
\widetilde\bX_i
:=
[\biota_T \;\; \bX_i],
\end{equation}
The individual-level projection matrix onto the span of the intercept and the regressors is
\begin{equation}
\label{eq:Pitilde}
\widetilde\bP_i
:=
\widetilde\bX_i(\widetilde\bX_i'\widetilde\bX_i)^{-1}\widetilde\bX_i',
\end{equation}
with corresponding annihilator
\begin{equation}
\label{eq:Mitilde}
\widetilde\bM_i
:=
\bI_T-\widetilde\bP_i.
\end{equation}
Unless stated otherwise, we assume that the required inverses exist with probability approaching one for all $i=1,\ldots,N$.

For the construction of the test statistic, in the baseline model define $T_a:=T-1$ and let
$\bB\in\mathbb R^{T_a\times T}$ have orthonormal rows spanning the subspace
orthogonal to $\biota_T$,
\begin{equation}
\label{eq:B_reduction}
\bB\bB'=\bI_{T_a},
\qquad
\bB\biota_T=\vzeros_{T_a},
\qquad
\bB'\bB
=
\bI_T-\frac{1}{T}\biota_T\biota_T'.
\end{equation}
Thus $\bB$ provides coordinates for the common residual time space. If additional deterministic regressors (linear trend, seasonal dummies, etc.) are included in every unit's residualization, $\bB$ and $T_a$ are defined analogously for the orthogonal complement of their common span.

Suppose we want to test whether the outcomes for unit $j$ appear in the equation for unit $i$. Under the null \eqref{eq:null}, the covariates included in the equation for unit $j$ may affect unit $j$, but they do not directly affect unit $i$. Similarly to \citet{depaula2025identifying}, this makes these covariates instruments that can be used to test whether $y_j$ appears in the equation for $y_i$. Consider then the outcomes for unit $i$ projected onto the covariates,
\begin{equation}
\label{eq:pi}
\bp_i:=\bB\bP_i \by_i.
\end{equation}
If unit $i$ and unit $j$ are unrelated, this projected signal should not correlate with the following residualized outcomes for unit $j$,
\begin{equation}
\label{eq:ai}
\ba_j:=\bB\widetilde\bM_j \by_j.
\end{equation}
Under the null hypothesis \eqref{eq:null}, using \eqref{eq:null_model},
\begin{equation}
\ba_i
=
\bB\widetilde\bM_i \bu_i,
\end{equation}
because \(\widetilde\bM_i\biota_T=\vzeros_{T}\) and \(\widetilde\bM_i\bX_i=\mZeros_{T\times p}\). Thus \(\ba_i\) is the residualized outcome of unit \(i\) under the null.

For each ordered pair \((i,j)\), \(i\neq j\), consider the scalar product $
\bp_j'\ba_i
$.
This quantity compares the part of unit \(j\)'s outcome that is predictable from unit \(j\)'s covariates with the part of unit \(i\)'s outcome that remains after removing unit \(i\)'s own covariates and fixed effect. Under the postulated $\mathbb{H}_{0}$ (and cross-sectional independence of the disturbances), there should be no systematic association between these two components for all \(i\neq j\).

The numerator of our test statistic is the sum over all such pairs
\begin{equation}
\label{eq:pairwise_score}
Q_N
=
\frac{1}{N}
\sum_{i=1}^N
\sum_{j\neq i}
\bp_j'\ba_i.
\end{equation}
As we show below, the normalization by \(N\) (and not by \(N(N-1)\)) is sufficient for the fixed-\(T\), large-\(N\) asymptotic framework we consider. Since \(\bp_i=\bB\bP_i \by_i\), \(\ba_i=\bB\widetilde\bM_i \by_i\), \(\bB'\bB = \bI_T-T^{-1}\biota_T\biota_T'\), and \(\bP_i\widetilde\bM_i=0\), we have
\begin{equation}
\bp_i'\ba_i
=
\by_i'\bP_i\bB'\bB\widetilde\bM_i \by_i=
\by_i'\bP_i\widetilde\bM_i \by_i
=
0.
\end{equation}
Therefore \eqref{eq:pairwise_score} can equivalently be written as
\begin{equation}
\label{eq:compact_score}
Q_N
=
\left(
\frac{1}{\sqrt N}\sum_{j=1}^N \bp_j
\right)'
\left(
\frac{1}{\sqrt N}\sum_{i=1}^N \ba_i
\right).
\end{equation}
This representation is useful both computationally and theoretically. It shows that the statistic depends on two (scaled) cross-sectional averages: the average projected outcome and the average residualized outcome.

Let
\begin{equation}
\widehat{\bh}_N
:=
\frac{1}{\sqrt N}\sum_{i=1}^N \bp_i,
\qquad
\widehat\bz_{a,N}
:=
\frac{1}{\sqrt N}\sum_{i=1}^N \ba_i.
\end{equation}
To standardize \(Q_N\), define
\begin{equation}
\label{eq:Omegaaa_hat}
\widehat\bOmega_{aa,N}
:=
\frac{1}{N}\sum_{i=1}^N
\ba_i\ba_i',
\qquad
\widehat\bGamma_N
:=
\frac{1}{N}\sum_{i=1}^N
\bp_i\ba_i'.
\end{equation}
The baseline statistic is then given by
\begin{equation}
\label{eq:RN_orig}
R_{N,\text{orig}}
=
\frac{
\widehat\bh_N'\widehat\bz_{a,N}
}{
\sqrt{
\operatorname{tr}\left(\widehat\bGamma_N^2\right)
+\widehat\bh_N'
\widehat\bOmega_{aa,N}
\widehat\bh_N
}
},
\end{equation}
whenever the denominator is positive.

Alternatively to the instrument-based interpretation given above, the statistic has a cross-sectional dependence (CD) test statistic interpretation. Standard CD test statistics (e.g. \citealp{breusch1980lagrange,pesaran2015testing}) often aggregate residual covariances (or correlations) across units, which would amount to analyzing the cross-product $\ba_j'\ba_i$. The present statistic instead aggregates covariances between projected outcomes and residualized outcomes. This distinction makes the test directional toward peer effects: it emphasizes cross-unit dependence that is related to observed covariates of potential peers. Dependence that is present only in residual-residual covariance is not the sole object of the test. Hence, unlike the said CD statistics, our approach is impossible without regressors \(\bX_i\).

In finite samples and in some covariance configurations, the term $\operatorname{tr}\left(\widehat\bGamma_N^2\right)$ appearing in the denominator need not be positive. We therefore also consider two alternative normalizations that guarantee positivity,
\begin{equation}
\label{eq:RN_plus}
R_{N,+}
:=
\frac{
\widehat\bh_N'\widehat\bz_{a,N}
}{
\sqrt{
\max\left\{\operatorname{tr}\left(\widehat\bGamma_{N}^2\right),0\right\}
+
\widehat\bh_N'\widehat\bOmega_{aa,N}\widehat\bh_N
}
},
\end{equation}
and
\begin{equation}
\label{eq:RN_sym}
R_{N,\mathrm{sym}}
:=
\frac{
\widehat\bh_N'\widehat\bz_{a,N}
}{
\sqrt{
\operatorname{tr}\left(\widehat\bGamma_{N}\widehat\bGamma_{N}'\right)
+
\widehat\bh_N'\widehat\bOmega_{aa,N}\widehat\bh_N
}
}.
\end{equation}
In the remainder, we occasionally write $R_{N,\star}$ where $\star\in\{\text{orig};+;\text{sym}\}$ to refer to the estimators jointly. The family of statistics $R_{N,\star}$ then takes the form
\begin{equation}
\label{eq:RN_star}
R_{N,\star}
:=
\frac{
\widehat\bh_N'\widehat\bz_{a,N}
}{
\sqrt{b_\star\left(\widehat\bGamma_{N}\right)+
\widehat\bh_N'\widehat\bOmega_{aa,N}\widehat\bh_N
}
},
\end{equation}
with the function $b_\star(\cdot)$ defined accordingly.

\section{Theoretical results}
\label{sec:theory}
\subsection{Asymptotic distribution}
\label{subsec:limit_distribution}

The results in this paper are derived under the following set of high-level regularity conditions.
\begin{assumption}
\label{ass:high_level}
The following conditions hold under the null hypothesis when $N\to\infty$:

\begin{enumerate}
\item[(i)] For some fixed \(\bh_0\in\mathbb R^{T_a}\) with \(\|\bh_0\|<\infty\)
\[
\begin{pmatrix}
\widehat\bh_N\\
\widehat\bz_{a,N}
\end{pmatrix}
\Rightarrow
\begin{pmatrix}
\bh\\
\bz_a
\end{pmatrix}
=
\begin{pmatrix}
\bh_0+\bz_p\\
\bz_a
\end{pmatrix},\qquad
\begin{pmatrix}
\bz_p\\
\bz_a
\end{pmatrix}
\sim
\mathcal{N} \left(
\mathbf \vzeros,
\begin{pmatrix}
\bOmega_{pp} & \bGamma\\
\bGamma' & \bOmega_{aa}
\end{pmatrix}
\right),
\]
where the variance covariance matrix is positive definite.
\item[(ii)] The second-moment estimators satisfy
$
\widehat\bOmega_{aa,N}\to_p \bOmega_{aa},
$ $\widehat\bGamma_N\to_p \bGamma.$

\item[(iii)] For the statistic under consideration, $
b_\star(\bGamma)+\bh'\bOmega_{aa}\bh>0$ $a.s.$

\end{enumerate}
\end{assumption}
Assumption~\ref{ass:high_level}(i) is the only central limit requirement used in the theory, where we adopt the local-to-zero asymptotic framework similar to weak IV asymptotic framework of \citet{staigerstock1997}. Assumption~\ref{ass:high_level}(ii) collects the probability limits of the \(N^{-1}\) objects. Assumption~\ref{ass:high_level}(iii) avoids singular cases and ensures that the conditional critical values are well defined.

In most cases \(\bGamma\) is given by the limiting covariance between the projected and residualized components:
\begin{equation}
\bGamma
=
\lim_{N\to\infty}
\frac{1}{N}\sum_{i=1}^N \E[\bp_i\ba_i'].
\end{equation}
 The value of \(\bGamma\) is important: it determines the dependence between the two Gaussian components \(\bz_p\) and \(\bz_a\). Nonzero covariance between these components is the source of the nonstandard conditional distributions studied below.

The following result follows immediately from the high-level convergence assumptions and the continuous mapping theorem.

\begin{theorem}
\label{thm:limit_statistic}
Suppose Assumption~\ref{ass:high_level} holds. Then, for each
\(\star\in\{\mathrm{orig};+;\mathrm{sym}\}\),
\begin{equation}
R_{N,\star}
\Rightarrow
R_\star(\bh_0),
\end{equation}
where
\begin{equation}
\label{eq:Rstar_limit}
R_\star(\bh_0)
=
\frac{
\bh'\bz_a
}{
\sqrt{
b_\star(\bGamma)+\bh'\bOmega_{aa}\bh
}
},
\qquad
\bh:=\bh_0+\bz_p.
\end{equation}
\end{theorem}

\begin{proof}
By Assumption~\ref{ass:high_level}(i), $
(\widehat\bh_N,\widehat\bz_{a,N})
\Rightarrow
(\bh,\bz_a).$
By Assumption~\ref{ass:high_level}(ii),
$
\widehat\bOmega_{aa,N}\to_p \bOmega_{aa}$ and
$\widehat\bGamma_N\to_p \bGamma.
$
Hence
$
\widehat\bh_N'\widehat\bz_{a,N}
\Rightarrow
\bh'\bz_a,
$
and
$
b_\star(\widehat\bGamma_N)
+
\widehat\bh_N'\widehat\bOmega_{aa,N}\widehat\bh_N
\Rightarrow
b_\star(\bGamma)+\bh'\bOmega_{aa}\bh.
$
The result follows by the continuous mapping theorem and the denominator positivity in Assumption~\ref{ass:high_level}(iii).
\end{proof}

The limit in \eqref{eq:Rstar_limit} is not generally standard normal. The numerator contains the product of two correlated Gaussian vectors, and the denominator depends on the random vector \(\bh\). When the deterministic component \(\bh_0\) is large, the linear term dominates and the usual Gaussian approximation emerges. In weak cases, however, the randomness of \(\bh\) is of first order, and the distribution can differ substantially from \(\mathcal{N}(0,1)\).

The following proposition formalizes the case where $\bh_0$ is large.
\begin{proposition}
\label{prop:large_h}
Suppose Assumption~\ref{ass:high_level} holds and let \(
\|\bh_0\|\to\infty\). Then
\begin{equation}
R_\star(\bh_0)
=
\frac{\bh_0'\bz_a}{\sqrt{\bh_0'\bOmega_{aa}\bh_0}}
+
o_p(1).
\end{equation}
Consequently, $R_\star(\bh_0)\Rightarrow \mathcal{N}(0,1)
\quad
\text{as } \|\bh_0\|\to\infty.$
\end{proposition}

\begin{proof}
Write \(\bh=\bh_0+\bz_p\). Since \(\bz_p=O_p(1)\) and \(\bz_a=O_p(1)\),
$
\bh'\bz_a
=
\bh_0'\bz_a+O_p(1).$
Similarly,
$
\bh'\bOmega_{aa}\bh
=
\bh_0'\bOmega_{aa}\bh_0+O_p(\|\bh_0\|),
$
and \(b_\star(\bGamma)=O(1)\). Therefore
$
\sqrt{b_\star(\bGamma)+\bh'\bOmega_{aa}\bh}
=
\sqrt{\bh_0'\bOmega_{aa}\bh_0}\{1+o_p(1)\}.
$
The first claim follows. Since \(\bh_0'\bz_a/\sqrt{\bh_0'\bOmega_{aa}\bh_0}\) is standard normal, the second claim follows.
\end{proof}

Proposition~\ref{prop:large_h} uses sequential limits where first $N\rightarrow\infty$ and then $\|\bh_0\|\rightarrow\infty$. It is well known in the panel data settings that such sequential limits may not be informative on the joint limit where $(N,\|\bh_0\|)\rightarrow\infty$ \citep{phillips1999linear}.

Sequential limit arguments can be avoided by assuming that
\begin{equation}
N^{\alpha-1/2}\widehat\bh_N=\bh_{0}+O_{P}(N^{\alpha-1/2}),
\end{equation}
for $\alpha\in[0,1/2)$. This assumption would be equivalent to the semi-strong IV setting as in \citet{ANTOINE2012350}. Under this assumption the conclusions of Proposition~\ref{prop:large_h} continue to hold.

An implication of Proposition~\ref{prop:large_h} and the subsequent discussion for general values of $\alpha$ is that the choice of \(b_\star(\cdot)\) is asymptotically irrelevant when $\alpha<1/2$. This implies that the original, positive-part, and symmetrized normalizations are asymptotically equivalent. The distinction between the three choices of \(b_\star(\cdot)\) matters only in weak projected-component regimes, where \(\widehat\bh_N=O_p(1)\) and the trace correction remains first order.

\subsection{Monte Carlo test}
\label{subsec:conditional_limit}

We now consider the case where \(\bh_0\) is moderate or close to zero. In that region, the statistic has a nonstandard distribution that depends on the nuisance parameters determining the local Gaussian experiment. Define
\begin{equation}
\label{eq:M_def}
\blm
:=
\bh-\bGamma\bOmega_{aa}^{-1}\bz_a.
\end{equation}
Since \(\bh=\bh_0+\bz_p\), the covariance between \(\bh\) and \(\bz_a\) is \(\bGamma\). Hence \(\blm\) is the residual from the linear projection of \(\bh\) on \(\bz_a\) in the Gaussian limit experiment.

\begin{theorem}
\label{thm:conditional_reduction}
Suppose Assumption~\ref{ass:high_level}(i)--(iii) hold. Then \(\blm\perp\bz_a\). Since
$
\bh:=\blm+\bGamma\bOmega_{aa}^{-1}\bz_a,
$
the distribution of \(R_\star(\bh_0)\) conditional on \(\blm\) is free of the local nuisance parameter \(\bh_0\).
\end{theorem}

\begin{proof}
Because \((\bh,\bz_a)\) is jointly Gaussian, it is enough to show that \(\blm\) and \(\bz_a\) are uncorrelated. We have
$
\operatorname{Cov}(\blm,\bz_a)
=
\operatorname{Cov}(\bh,\bz_a)
-
\bGamma\bOmega_{aa}^{-1}\operatorname{Var}(\bz_a)
=
\bGamma-\bGamma\bOmega_{aa}^{-1}\bOmega_{aa}
=
\vzeros.$
Thus \(\blm\perp\bz_a\). Conditional on \(\blm=\blm_0\), and using \eqref{eq:M_def}, the statistic can be written as
\[
R_\star(\bh_0)
=R_\star(\blm_0)
=
\frac{
\{\blm_0+\bGamma\bOmega_{aa}^{-1}\bz_a\}'\bz_a
}{
\sqrt{
b_\star(\bGamma)
+
\{\blm_0+\bGamma\bOmega_{aa}^{-1}\bz_a\}'\bOmega_{aa}
\{\blm_0+\bGamma\bOmega_{aa}^{-1}\bz_a\}
}
},
\]
where \(\bz_a\mid \blm=\blm_0\sim \mathcal{N}(\mathbf 0,\bOmega_{aa})\).
\end{proof}

Theorem~\ref{thm:conditional_reduction} is the main theoretical result for feasible inference using our approach. The limiting statistic \eqref{eq:Rstar_limit} may have a complicated unconditional distribution, but the conditional distribution (given \(\blm\)) is free of the local nuisance parameter $\bh_0$. This suggests using simulated conditional critical values rather than standard Gaussian critical values. The conditional test statistic in Theorem~\ref{thm:conditional_reduction} is infeasible because it depends on the unknown matrices \(\bGamma\) and \(\bOmega_{aa}\). A feasible version is provided in Algorithm \ref{alg:feasible_plugin}.

\begin{algorithm}[htbp]
\caption{Feasible plug-in conditional Monte Carlo test}
\label{alg:feasible_plugin}
\begin{algorithmic}[1]
\State Compute the conditioning statistic:
\[
\widehat\blm_N
:=
\widehat\bh_N
-
\widehat\bGamma_N
\widehat\bOmega_{aa,N}^{-1}
\widehat\bz_{a,N}.
\]

\For{$b=1,\ldots,B$}
    \State Draw independently
    \[
    \bz_{a,b}^\ast \sim \mathcal{N}(\vzeros,\widehat\bOmega_{aa,N}).
    \]
    \State Set
    \[
    \bh_b^\ast
    :=
    \widehat\blm_N
    +
    \widehat\bGamma_N\widehat\bOmega_{aa,N}^{-1}\bz_{a,b}^\ast.
    \]
    \State Compute
    \[
    R_{\star,b}^\ast
    =
    \frac{
    (\bh_b^\ast)'\bz_{a,b}^\ast
    }{
    \sqrt{
    b_\star(\widehat\bGamma_N)
    +
    (\bh_b^\ast)'\widehat\bOmega_{aa,N}\bh_b^\ast
    }
    }.
    \]
\EndFor

\State Let \(\widehat c_{\alpha,\star}(\widehat\blm_N)\) be the empirical \((1-\alpha)\)-quantile of
\[
\{|R_{\star,b}^\ast|: b=1,\ldots,B\}.
\]
\If{$|R_{N,\star}|>\widehat c_{\alpha,\star}(\widehat\blm_N)$}
    \State Reject $H_{0}$.
\Else
    \State Do not reject $H_{0}$.
\EndIf
\end{algorithmic}
\end{algorithm}
The validity of the test conducted through this algorithm is given by the following theorem.
\begin{theorem}
\label{thm:feasible_cmc}
Suppose Assumption~\ref{ass:high_level} holds and \(B=B_N\to\infty\). Then, for each fixed \(\bh_0\in\mathbb R^{T_a}\) and any $\alpha \in (0,1)$,
\begin{equation}
\limsup_{N\to\infty}
P_{\bh_0}\left(
|R_{N,\star}|>\widehat c_{\alpha,\star}(\widehat\blm_N)
\right)
= \alpha.
\end{equation}
If Assumption~\ref{ass:high_level} also holds uniformly over a set
\(\mathcal H\subseteq\mathbb R^{T_a}\), then
\begin{equation}
\limsup_{N\to\infty}
\sup_{\bh_0\in\mathcal H}
P_{\bh_0}\left(
|R_{N,\star}|>\widehat c_{\alpha,\star}(\widehat\blm_N)
\right)
= \alpha.
\end{equation}
\end{theorem}

\begin{proof}
By Assumption~\ref{ass:high_level},
$
(\widehat\bh_N,\widehat\bz_{a,N},
\widehat\bGamma_N,\widehat\bOmega_{aa,N})
\Rightarrow
(\bh,\bz_a,\bGamma,\bOmega_{aa}).
$
Therefore, $
\widehat\blm_N
=
\widehat\bh_N
-
\widehat\bGamma_N
\widehat\bOmega_{aa,N}^{-1}
\widehat\bz_{a,N}
\Rightarrow
\blm
$.
The simulated conditional law is the plug-in analogue of the conditional law in Theorem~\ref{thm:conditional_reduction}. Since
$
\widehat\bGamma_N\to_p\bGamma,
$, $
\widehat\bOmega_{aa,N}\to_p\bOmega_{aa}
$, and $\widehat\blm_N\Rightarrow \blm,
$
and since the denominator is positive in the limit, the simulated distribution converges weakly to the correct conditional distribution. The empirical quantile converges to the corresponding population conditional quantile because \(B_N\to\infty\) and the conditional distribution is continuous at all quantiles since its a smooth function of a Gaussian vector. The pointwise size result follows from Theorem~\ref{thm:conditional_reduction}. The uniform result follows by the same argument under uniform convergence.
\end{proof}

\begin{remark}
Note that positivity of the denominator of the test is a crucial requirement for the asymptotic approximation, and therefore the exact validity of the test. Even in the asymptotic limit, the positivity requirement may fail for the original statistic. If in this case, we append the test with the decision rule not to reject when positivity fails, the test will become conservative and the equality in \eqref{thm:feasible_cmc} should be replaced with an inequality.
\end{remark}
When \(\bGamma=\mathbf 0\), the conditional reduction becomes trivial. In that case, \(\blm=\bh\), and conditional on \(\blm\),
\begin{equation}
R_\star(\bh_0)
=
\frac{\blm'\bz_a}{\sqrt{\blm'\bOmega_{aa}\blm}}.
\end{equation}
Hence the conditional distribution is standard normal whenever \(\blm'\bOmega_{aa}\blm>0\) as formalized in the corollary below.

\begin{corollary}
\label{cor:gamma_zero}
Let $\|\blm\|>0$. If \(\bGamma=\mathbf 0\), then
$
R_\star(\bh_0)\mid \blm
\sim \mathcal{N}(0,1).
$
The conditional critical value equals the usual two-sided Gaussian critical value,
$
c_{\alpha,\star}(\blm)=z_{1-\alpha/2}.
$
\end{corollary}

\subsection{Relation to identification-robust inference in linear IV}
An important observation is that the limiting experiment that arises for the proposed test statistic is identical to that of the score in the linear IV model. Consider for ease of exposition a linear IV model with single endogenous regressor and instruments \(\bz_i\in\mathbb R^K\),
\begin{equation}
y_i=x_i\beta+u_i,\qquad x_i=\bz_i'\bpi_N + v_i.
\end{equation}
For testing \(H_0:\beta=\beta_0\),
define the sample moment
\begin{equation}
\bg_N(b)
=
\frac{1}{\sqrt N}\sum_{i=1}^N \bz_i(y_i-x_i b)
=
\frac{1}{\sqrt N}\bZ'(\by-\bx b).
\end{equation}
This corresponds to the residualized vector $\widehat{\bz}_{a,N}$.
The (negative) sample Jacobian of the moment with respect to \(\beta\) is
\begin{equation}
\bj_N:=-
\frac{\partial \bg_N(b)}{\partial b}\bigg|_{b=\beta_0}
=
\frac{1}{\sqrt N}\bZ'\bx.
\end{equation}
Under weak IV asymptotics, $\bpi_N = \bpi_0/\sqrt{N}$, so $\bH_N=\bQ_N\bpi_0 + \frac{1}{\sqrt{N}}\bZ'\bv$, for $\bQ_N:=\frac1N \bZ'\bZ.$  Hence, $\bQ_N^{-1}\bj_N = \bpi_0 + \frac{1}{\sqrt{N}}\bQ_N^{-1}\bZ'\bv$, which corresponds to the projected outcomes in $\widehat{\bh}_{N}$.

Analogous to Assumption~\ref{ass:high_level}(i), suppose that, under \(\mathbb{H}_0\),
\[
\begin{pmatrix}
\bj_N\\
\bg_N(\beta_{0})
\end{pmatrix}
\Rightarrow
\begin{pmatrix}
\bj\\
\bg
\end{pmatrix}
=
\begin{pmatrix}
\bj_0+\bk\\
\bg
\end{pmatrix},
\qquad \begin{pmatrix}
\bk\\
\bg
\end{pmatrix}
\sim
\mathcal{N}\left(
\bzero,
\begin{pmatrix}
\bOmega_{jj} & \bGamma_{jg}\\
\bGamma_{jg}' & \bOmega_{gg}
\end{pmatrix}
\right),
\]
 where the variance-covariance matrix again is positive definite.
Define
\begin{equation}
\bd
:=
\bj-\bGamma_{jg}\bOmega_{gg}^{-1}\bg.
\end{equation}
Then \(\bd\perp \bg\)
in the Gaussian limit experiment, just like $\blm\perp \bz_a$ in the network case. The standard normal behavior of the conditional test with $\bGamma=\bO$ is now completely expected, as this corresponds to the IV case where there is no endogeneity.

Evidently, the analogy with IV framework, immediately suggests a different network statistic that is pivotal. For example, the score statistic of \citet{kleibergen2002pivotal} in our context is given by
\begin{equation}
K_{N,\text{network}} = \left(\frac{\widehat\blm_N'\widehat{\bOmega}_{aa,N}^{-1}\widehat{\bz}_{a,N}}{\sqrt{\widehat\blm_{N}'\widehat{\bOmega}_{aa,N}^{-1}\widehat\blm_{N}}}\right)^{2}.
\end{equation}
Such that under Assumption \ref{ass:high_level}
\begin{equation}
K_{N,\text{network}}\Rightarrow\chi^2_1.
\end{equation}

We take no stance on whether the conditional (also related to \citealp{moreira2003conditional}) or the score test should be preferred. The pivotal nature of the score is attractive, but in terms of interpretability of the test statistic in a network setting, the original version is likely preferred. In terms of power, neither dominates the other as we show in a limiting experiment in Appendix~\ref{app:compareCondKlei}.
\subsection{Numerical illustrations for \texorpdfstring{$T_a=2$}{Ta=2}}
\label{subsec:canonical_2x2}

The conditional representation is exact for any fixed \(T_a\), but its implications are easiest to see when \(T_a=2\). In this subsection we assume \(T_a=2\) and \(\bOmega_{pp}=\bOmega_{aa}=\bI_2\) such that the covariance matrix \(\bGamma\) solely determines the geometry of the conditional distribution. In what follows, we consider two specifications for \(\bGamma\) that best illustrate features of our proposed conditional approach.

\begin{figure}[t]
\caption{Skew \(2\times2\) design.}
\label{fig:skew_results}
\centering
\begin{subfigure}{0.48\textwidth}
    \caption{Conditional critical values versus \(m=\|\blm\|\).}
    \label{fig:skew_cv}
    \centering
    \includegraphics[width=\linewidth]{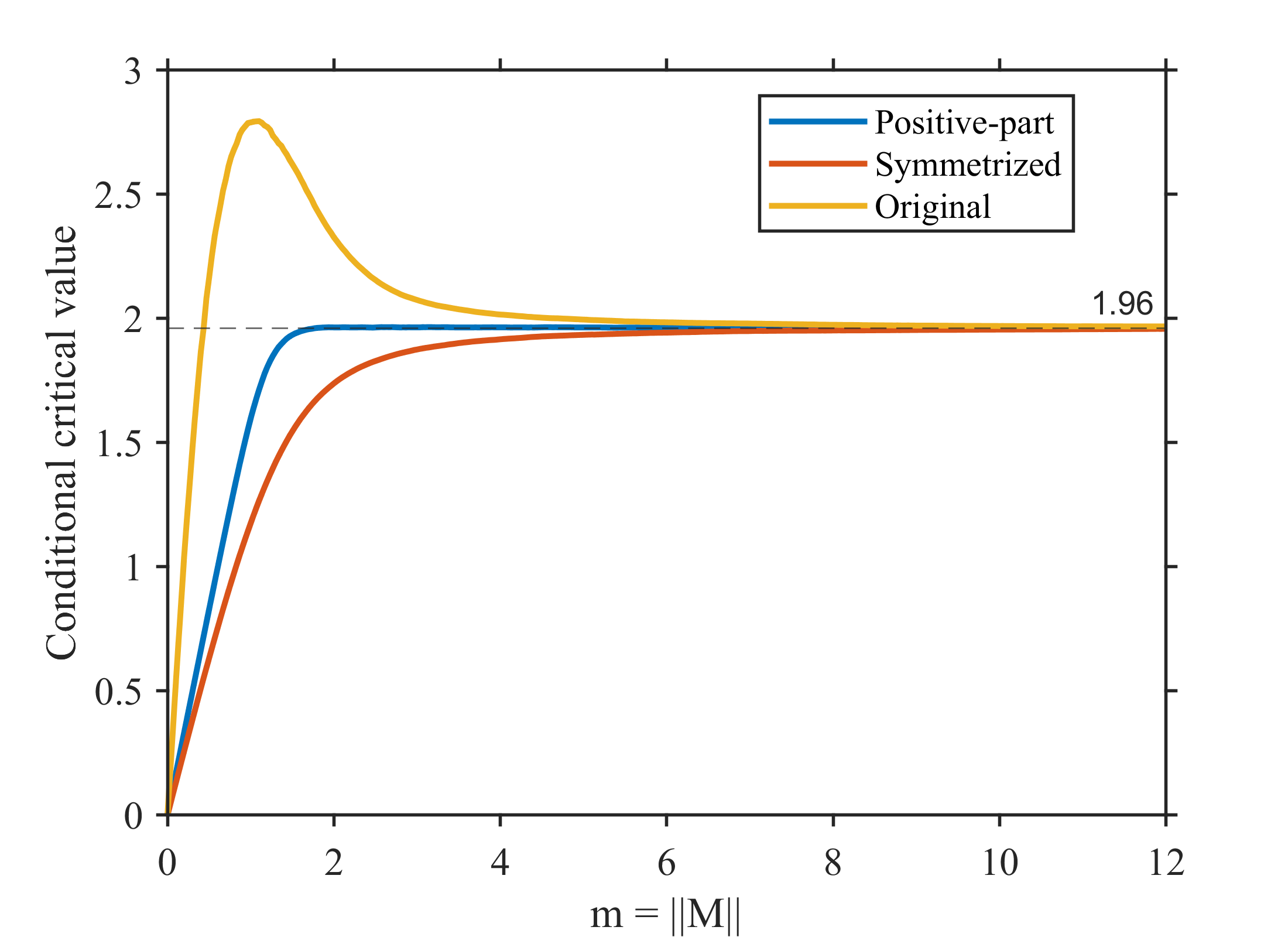}
\end{subfigure}
\hfill
\begin{subfigure}{0.48\textwidth}
    \caption{Worst-case rejection over the \(\bh_0\)-grid.}
    \label{fig:skew_worst}
    \centering
    \includegraphics[width=\linewidth]{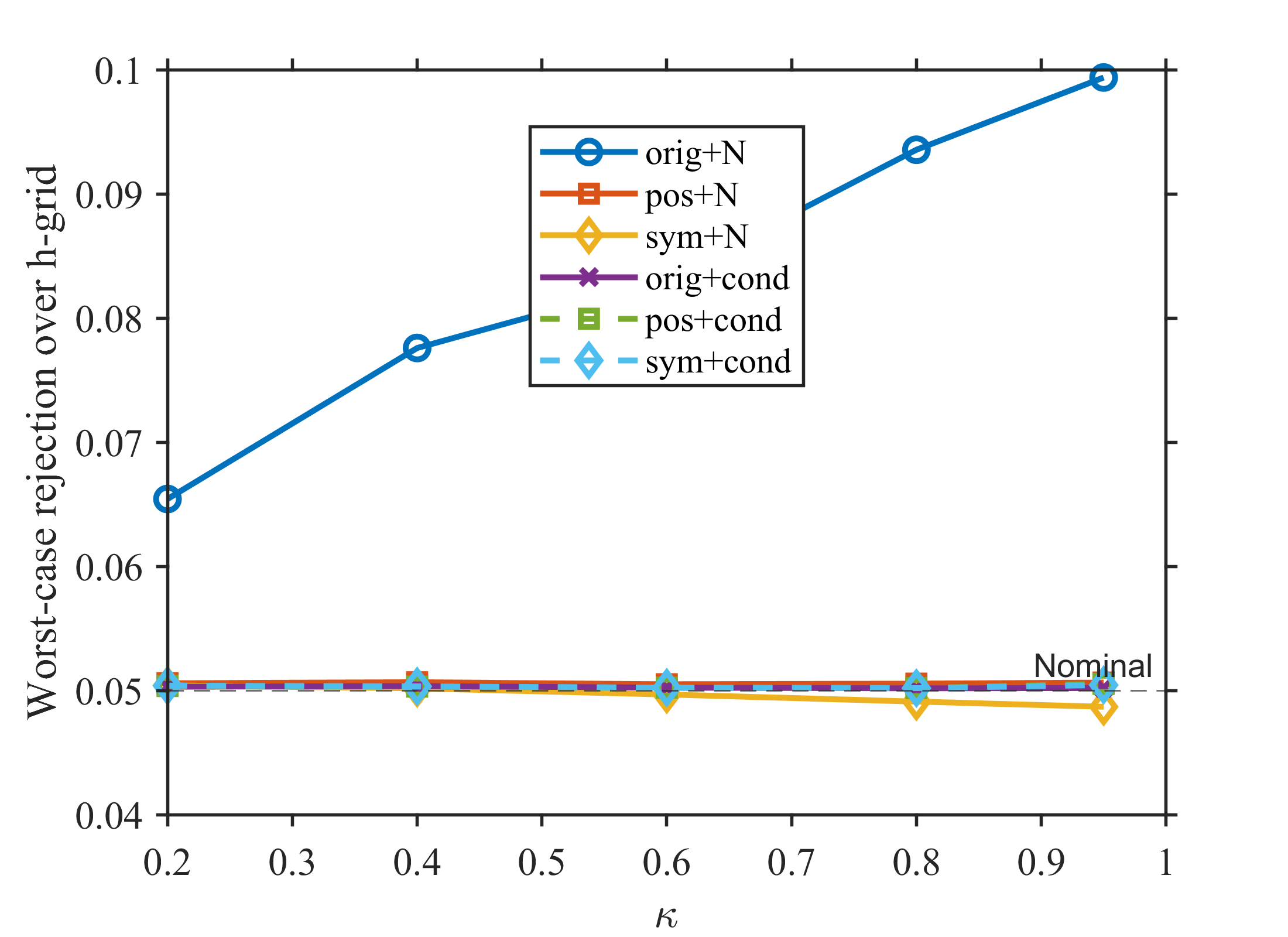}
\end{subfigure}
\begin{minipage}{0.92\textwidth}
\small \textit{Note}: The figure reports results for the skew covariance design \(\bGamma=\kappa\bK\). The conditional critical value depends only on \(m=\|\blm\|\), starts at zero when \(m=0\), and approaches the Gaussian benchmark as \(m\) increases. The rejection plot reports worst-case rejection probabilities over the grid for \(\bh_0\).
\end{minipage}
\end{figure}
\paragraph{Skew case.}
Suppose
\begin{equation}
\bGamma=\kappa\bK,
\qquad
\bK=
\begin{pmatrix}
0&1\\
-1&0
\end{pmatrix}.
\end{equation}
Because \(\bK\) is invariant under rotations, the conditional law depends on \(\blm\) only through \(m=\|\blm\|\). Rotate coordinates so that \(\blm=(m,0)'\), and write \(\bz_a=(z_{1},z_{2})'\). Then
\begin{equation}
\bh
=
\blm+\kappa\bK\bz_a
=
\begin{pmatrix}
m+\kappa z_{2}\\
-\kappa z_{1}
\end{pmatrix}.
\end{equation}
Therefore, $
\bh'\bz_a
=
(m+\kappa z_{2})z_{1}-\kappa z_{1}z_{2}
=
mX.
$
Thus the skew component drops out of the numerator. It affects the statistic only through the denominator.

For the three normalizations,
$
b_{\mathrm{orig}}(\bGamma)
=
\operatorname{tr}(\bGamma^2)
=
-2\kappa^2,$ $
b_+(\bGamma)=0,$
and $
b_{\mathrm{sym}}(\bGamma)
=
\operatorname{tr}(\bGamma\bGamma')
=
2\kappa^2.
$
We can conclude that the original, positive-part, and symmetrized statistics can behave differently in skew designs. In particular, the original denominator may be too small, while the symmetrized denominator can be substantially larger.

At \(m=0\), the numerator is identically zero, i.e., \(\bh'\bz_a=0\) a.s. Hence the exact conditional critical value collapses to zero for the positive-part and symmetrized statistics. As \(m\) increases, the conditional distribution moves back toward the standard Gaussian benchmark.

To illustrate these arguments numerically, Figure~\ref{fig:skew_results} reports the conditional critical value as a function of \(m\) and the worst-case rejection probability over the grid for \(\bh_0\). The conditional critical value starts at zero and approaches the Gaussian benchmark as \(m\) increases. The original Gaussian-calibrated statistic overrejects as \(\kappa\) grows, while the conditionally calibrated positive-part and symmetrized statistics remain close to nominal size.

\begin{figure}[t]
\caption{Symmetric trace-zero \(2\times2\) design.}
\label{fig:symmtrace_results}
\centering
\begin{subfigure}{0.48\textwidth}
     \caption{Slices at \(\psi=0\) and \(\psi=\pi/4\).}
    \label{fig:symmtrace_slices}
   \centering
    \includegraphics[width=\linewidth]{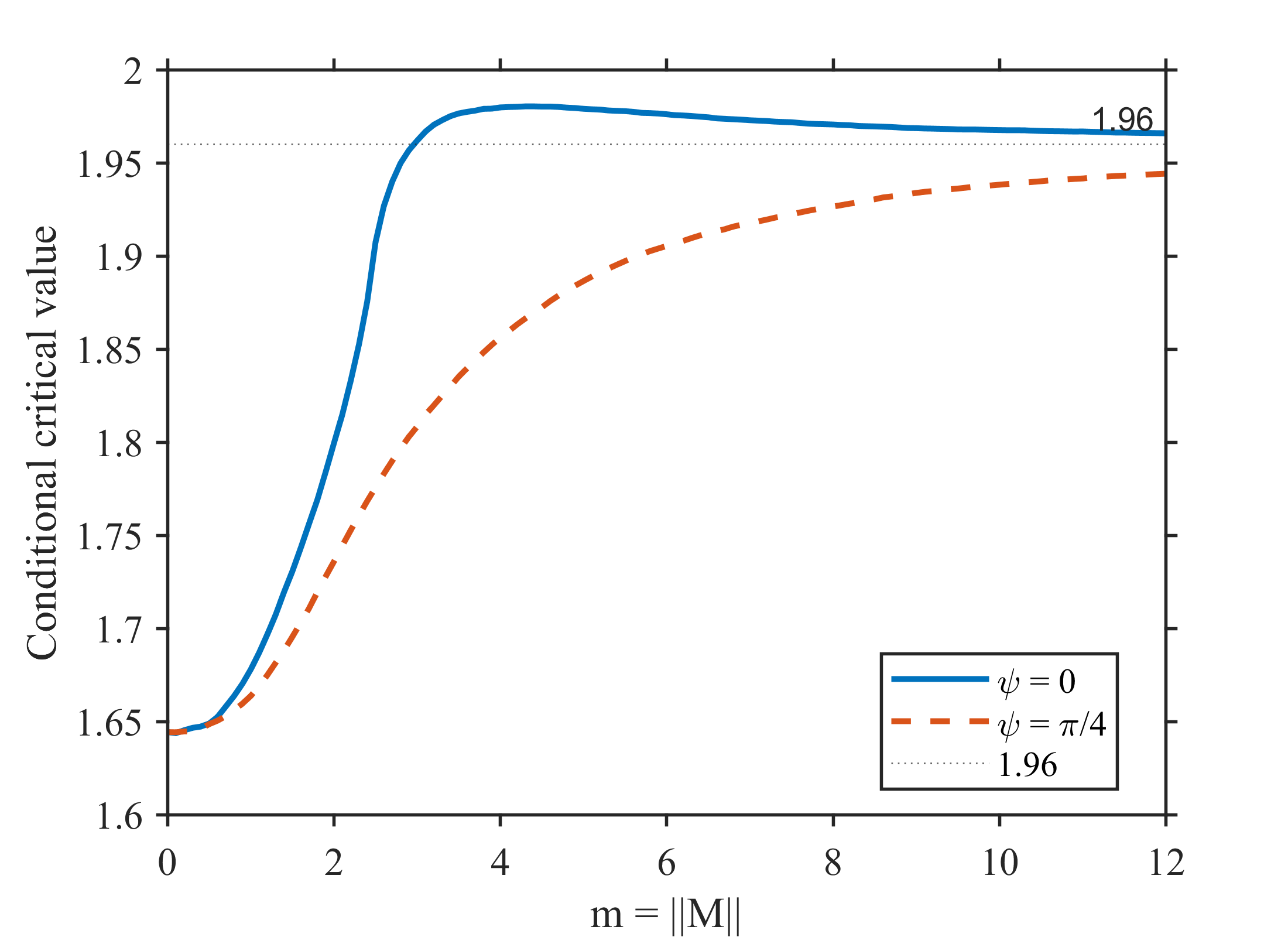}
\end{subfigure}
\hfill
\begin{subfigure}{0.48\textwidth}
     \caption{Worst-case rejection over the \(\bh_0\)-grid.}
    \label{fig:symmtrace_worst}
   \centering
    \includegraphics[width=\linewidth]{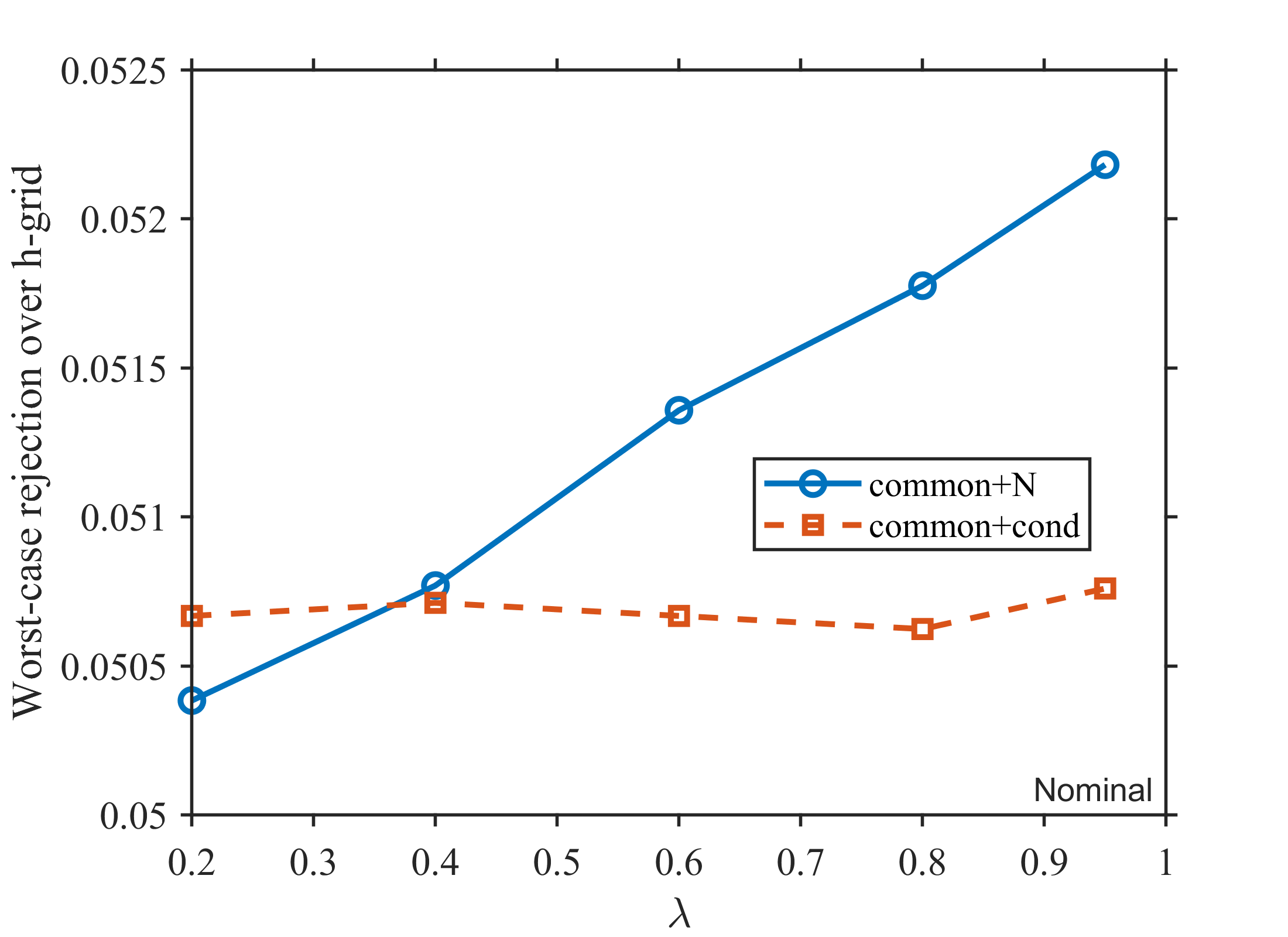}
\end{subfigure}
\begin{minipage}{0.92\textwidth}
\small \textit{Note}: The figure reports results for the symmetric trace-zero design. The conditional critical value depends on both \(m=\|\blm\|\) and the orientation \(\psi\) of \(\blm\) relative to the eigenbasis of \(\bGamma\).
\end{minipage}
\end{figure}

\paragraph{Symmetric trace-zero case.}
Suppose
\begin{equation}
\bGamma=\lambda\bD_\varphi,
\qquad
\bD_\varphi
=
\bR(\varphi)
\begin{pmatrix}
1&0\\
0&-1
\end{pmatrix}
\bR(\varphi)',\qquad  \bR(\varphi)
=
\begin{pmatrix}
\cos\varphi & -\sin\varphi\\
\sin\varphi & \cos\varphi
\end{pmatrix}.
\end{equation}
Note that in this example \(\bGamma\) is symmetric and
$
\operatorname{tr}(\bGamma^2)
=
\operatorname{tr}(\bGamma\bGamma')
=
2\lambda^2.
$
Thus
$
b_{\mathrm{orig}}(\bGamma)
=
b_+(\bGamma)
=
b_{\mathrm{sym}}(\bGamma)
=
2\lambda^2.
$
Therefore, in this design the three denominator corrections coincide.

Let \(\bu_1=(\cos \phi, \sin \phi)'\) and \(\bu_2=(-\sin \phi,\cos \phi)'\) denote the eigenvectors of \(\bGamma\), so that
\(
\bGamma \bu_1=\lambda \bu_1
\)
and
\(
\bGamma \bu_2=-\lambda \bu_2.
\)
Write
\(
\blm = m(\cos\psi\,\bu_1+\sin\psi\,\bu_2)\), and
\(m=\|\blm\|.
\)
Then \(\psi\) is the angle of \(\blm\) relative to the eigenbasis of \(\bGamma\).
 Conditional on \(\blm\),
\(
\bh'\bz_a
=
\blm'\bz_a+\bz_a'\bGamma\bz_a = m(\cos\psi z_1 + \sin\psi z_2) + \lambda(z_1^2-z_2^2).
\)
When \(m=\|\blm\|\) is small, the quadratic term \(\bz_a'\bGamma\bz_a\) is important. When \(m\) is large, the linear term \(\blm'\bz_a\) dominates. The transition between these two regimes depends on the angle \(\psi\) between \(\blm\) and the eigenbasis of \(\bGamma\). Thus the conditional critical value is generally a function of both \(m\) and \(\psi\).
The left panel of Figure~\ref{fig:symmtrace_results} shows this transition.

The conditional critical value can be nonmonotone in \(m\), reflecting the movement from a quadratic-form regime to a regular Gaussian regime. As the conditional critical values are lower then the standard normal critical values, we expect worse case rejection over $\bh_0$ to only be mildly above the nominal level. Indeed, the worst-case rejection rates shown in the right-panel over a grid of $\lambda$ shows that size distortions are very mild.

\section{Extension: Time effects}
\subsection{Weighted test statistic}
\label{subsec:time_effects_weighted}
Our analysis so far has assumed that under $\mathbb{H}_{0}$ all cross-sectional units are independent. This is generally a  restrictive assumption that is difficult to satisfy in practice. In empirical analysis, it is customary to saturate the model with additive common time components \(\bd\in\mathbb R^T\).

In the presence of common time effects, the model in $\mathbb{H}_{0}$ takes the form
\begin{equation}
\by_i=\bX_i\bbeta_i+\lambda_i\biota_T+\bd+\bu_i.
\end{equation}
As a result,
\begin{equation}
\ba_i=\bB\widetilde\bM_i \by_i
=\bB
\widetilde\bM_i \bd+\bB\widetilde\bM_i \bu_i
\end{equation}
is no longer centered at zero. Consequently, the unweighted aggregate $
\sum_{i=1}^N \ba_i$
can contain a common time-effect component and the proposed tests are not valid.

Empirical researchers usually simply treat time effects as parameters to be estimated (jointly with other parameters). Unfortunately, for your purpose this is not desirable, as estimating the time effects creates cross-sectional dependence in the residuals, see e.g., \citet{juodis2022incidental}. We therefore consider a weighted variant of the statistic that remove the time effects on the \(\ba_i\)-side of \eqref{eq:compact_score}. Let
\begin{equation}
Q_{w,N}
=
\left(\frac{1}{\sqrt{N}}\sum_{j=1}^N \bp_j\right)'
\left(\frac{1}{\sqrt{N}}\sum_{i=1}^N w_i \ba_i\right),
\end{equation}
where the weights \(w_i\) for $i=1,\ldots,N$ are assumed to be measurable with respect to the regressors only. The purpose of the weights is to eliminate common time effects from the aggregated residualized component. We impose the exact balancing condition\footnote{This restriction is equivalent to $\sum_{i=1}^N w_i\bB\widetilde\bM_i\bB'=\mZeros$.}
\begin{equation}
\label{eq:exact_a_balance}
\sum_{i=1}^N w_i\bB\widetilde\bM_i=\mZeros.
\end{equation}
Then, for every deterministic common time vector \(\bd\in\mathbb R^T\),
\begin{equation}
\sum_{i=1}^N w_i\bB\widetilde\bM_i \bd=\vzeros_{T_a}.
\end{equation}
Thus the weighted aggregate \(\sum_i w_i \ba_i\) is free of any additive common time effect under $\mathbb{H}_{0}$. A similar idea is employed by \citet{juodis2022incidental}, who choose the weights to be Rademacher random variables. In this paper, we leverage on the availability of strictly exogenous regressors, that allows for a construction of the weights fully described by these regressors. Define $
\widehat\bz_{a,w,N}
:=
\frac{1}{\sqrt N}\sum_{i=1}^N w_i \ba_i.
$
The numerator of the weighted test statistic is now
\begin{equation}
  Q_{w,N}=
\widehat\bh_N'\widehat\bz_{a,w,N}.
\end{equation}

Using the equivalent reduced-space restriction, define
$\bL_i:=\bB\widetilde\bM_i\bB'$. Since $\bL_i$ is symmetric, the restriction in \eqref{eq:exact_a_balance} can be written as a finite-dimensional linear system. Define
\begin{equation}
\bC:=
\bigl[
\operatorname{vech}(\bL_1),
\ldots,
\operatorname{vech}(\bL_N)
\bigr],
\end{equation}
where $\operatorname{vech}(\cdot)$ is the half-vectorization operator, see e.g., \citet{magnus}. Then \eqref{eq:exact_a_balance} is equivalent to $\bC\bw=\vzeros_{T_a(T_a+1)/2}$ for $\bw:=(w_{1},\ldots,w_{N})'$. Now let
\begin{equation}\label{eq:s}
\bs:=(s_1,\ldots,s_N)',
\end{equation}
be a regressor-based score chosen to proxy exposure to the alternative. We choose \(\bw\) as the normalized projection of \(\bs\) onto \(\operatorname{Null}(\bC)\). This enforces exact balancing by construction, while the score selects a direction within the balancing null space and therefore affects power. Furthermore, we impose
\begin{equation}
\frac{1}{N}\sum_{i=1}^N w_i^2=1,
\end{equation}
as a normalization.\footnote{\citet{Hsiao02102022} also put forward estimation method that makes use of unit-weights \(\bw\). Our approach differs from theirs as we do not estimate weights jointly with $\bbeta_{i}$. As a result, our approach is limited to additive time effects, while their approach also allows for interactive fixed effects.}

The denominator of the weighted statistic is based on
\begin{equation}
\widehat\bGamma_{w,N}
:=
\frac{1}{N}\sum_{i=1}^N \bp_i \ba_i^{\circ\prime},
\qquad
\widehat\bOmega_{aa,w,N}
:=
\frac{1}{N}\sum_{i=1}^N \ba_i^\circ \ba_i^{\circ\prime}.
\end{equation}
Here $\ba_i^{\circ}$ are adjusted (residualized) versions of $\ba_i$ that take into account that while the exact balancing condition removes common time effects from the aggregated numerator, it does not by itself remove their contribution from the quadratic objects used in the denominator. We discuss this adjustment in Section \ref{subsec:denominator} in detail.

For \(\star\in\{\mathrm{orig};+;\mathrm{sym}\}\), define
\begin{equation}
R_{N,a,\star}
:=
\frac{
\widehat\bh_N'\widehat\bz_{a,w,N}
}{
\sqrt{
b_\star(\widehat\bGamma_{w,N})
+
\widehat\bh_N'\widehat\bOmega_{aa,w,N}\widehat\bh_N
}
}.
\end{equation}
The statistic \(R_{N,a,\star}\) is the \(a\)-side weighted analogue of \(R_{N,\star}\). It is designed for settings in which unrestricted common time effects are a concern. The unweighted statistic remains the baseline procedure when such time effects are absent or explicitly modeled by observed time controls.
\subsection{Asymptotic analysis}
\label{subsec:weighted_theory}

We now state the corresponding high-level result for the weighted statistic. The same conditional argument applies once the time-effect nuisance has been removed from the weighted aggregate and from the denominator objects.

The following assumption is the weighted analogue of Assumption~\ref{ass:high_level}.

\begin{assumption}
\label{ass:weighted_high_level}
Under the null with an arbitrary additive common time effect \(\bd\in\mathbb R^T\), the following conditions hold.

\begin{enumerate}
\item[(i)] The weights \(\bw \in\mathbb R^N\) are measurable with respect to the regressors and satisfy
$
\sum_{i=1}^N w_i\bB\widetilde\bM_i=0,$ and
$ \max_{1\le i\le N}|w_i|=o_p(\sqrt N).$

\item[(ii)] For each fixed \(\bh_0\in\mathbb R^{T_a}\),
\[
\begin{pmatrix}
\widehat\bh_N\\
\widehat\bz_{a,w,N}
\end{pmatrix}
\Rightarrow
\begin{pmatrix}
\bh\\
\bz_{a,w}
\end{pmatrix}
=
\begin{pmatrix}
\bh_0+\bz_p\\
\bz_{a,w}
\end{pmatrix},\qquad
\begin{pmatrix}
\bz_p\\
\bz_{a,w}
\end{pmatrix}
\sim
\mathcal{N}\left(
\mathbf 0,
\begin{pmatrix}
\bOmega_{pp,w} & \bGamma_w\\
\bGamma_w' & \bOmega_{aa,w}
\end{pmatrix}
\right),
\]
where the variance-covariance matrix is positive definite.

\item[(iii)] The adjusted denominator objects satisfy
$
\widehat\bGamma_{w,N}\to_p\bGamma_w,$ $
\widehat\bOmega_{aa,w,N}\to_p\bOmega_{aa,w}.$
\item[(iv)] For the statistic under consideration, $
b_\star(\bGamma_w)+\bh'\bOmega_{aa,w}\bh>0$ $a.s.$
\end{enumerate}
\end{assumption}

The exact balancing condition in Assumption~\ref{ass:weighted_high_level}(i) removes the common time effect from the aggregated weighted \(a\)-component:
\begin{equation}
\sum_{i=1}^N w_i\ba_i
=
\sum_{i=1}^N w_i\bB\widetilde\bM_i\bu_i
\end{equation}
under the null.

\begin{theorem}
\label{thm:weighted_limit}
Suppose Assumption~\ref{ass:weighted_high_level} holds. For each
\(\star\in\{\mathrm{orig};+;\mathrm{sym}\}\),
\begin{equation}
R_{N,a,\star}
\Rightarrow
R_{a,\star}(\bh_0),
\end{equation} where
\begin{equation}
R_{a,\star}(\bh_0)
:=
\frac{
\bh'\bz_{a,w}
}{
\sqrt{
b_\star(\bGamma_w)+\bh'\bOmega_{aa,w}\bh
}
}.
\end{equation}
Moreover, defining $
\blm_a
:=
\bh-\bGamma_w\bOmega_{aa,w}^{-1}\bz_{a,w},
$
we have \(\blm_a\perp\bz_{a,w}\). Conditional on \(\blm_a\), the distribution of \(R_{a,\star}(\bh_0)\) is free of the local nuisance parameter \(\bh_0\).
\end{theorem}

The feasible conditional Monte Carlo implementation is the direct weighted analogue of the unweighted procedure. Define
\begin{equation}
\widehat\blm_{N,a}
:=
\widehat\bh_N
-
\widehat\bGamma_{w,N}
\widehat\bOmega_{aa,w,N}^{-1}
\widehat\bz_{a,w,N}.
\end{equation}
Conditional on the sample, simulate
\begin{equation}
\bz_a^\ast\sim\mathcal{N}(\vzeros,\widehat\bOmega_{aa,w,N}),
\qquad
\bh^\ast
:=
\widehat\blm_{N,a}
+
\widehat\bGamma_{w,N}
\widehat\bOmega_{aa,w,N}^{-1}\bz_a^\ast,
\end{equation}
and compute
\begin{equation}
R_{a,\star}^\ast
=
\frac{
(\bh^\ast)'\bz_a^\ast
}{
\sqrt{
b_\star(\widehat\bGamma_{w,N})
+
(\bh^\ast)'\widehat\bOmega_{aa,w,N}\bh^\ast
}
}.
\end{equation}
The empirical \((1-\alpha)\)-quantile of \(|R_{a,\star}^\ast|\) is used as the conditional critical value. Under Assumption~\ref{ass:weighted_high_level} and with the number of simulation draws diverging, the resulting weighted conditional test has asymptotic size at most \(\alpha\).

This result shows that the weighted statistic does not require a new inferential principle. It is the same conditional Gaussian limit experiment applied to a nuisance-free weighted aggregate. The only additional step is the construction of weights and adjusted denominator objects that remove common time effects.
\subsection{Denominator adjustment}
\label{subsec:denominator}
Below we discuss how the denominator of our test statistic should be adjusted to account for the presence of time effects in the model.

Consider the naive estimator
\begin{equation}
\widehat\bOmega_{aa,w,N}
=
\frac{1}{N}\sum_{i=1}^N w_{i}^{2}\ba_i \ba_i^{\prime}.
\end{equation}
By construction
\begin{align}
\widehat\bOmega_{aa,w,N}&=\frac{1}{N}\sum_{i=1}^N w_{i}^{2}\bB\widetilde{\bM}_{i}\bu_i\bu_i'\widetilde{\bM}_i\bB'\notag\\
&+\frac{1}{N}\sum_{i=1}^N w_{i}^{2}\bB\widetilde{\bM}_{i}\bd\bu_i'\widetilde{\bM}_i\bB'+\frac{1}{N}\sum_{i=1}^N w_{i}^{2}\bB\widetilde{\bM}_{i}\bu_i\bd'\widetilde{\bM}_i\bB
'\notag\\
&+\frac{1}{N}\sum_{i=1}^N w_{i}^{2}\bB\widetilde{\bM}_{i}\bd\bd'\widetilde{\bM}_i\bB'.
\end{align}
The term in the first line is the only one that asymptotically determines the variance of $\widehat\bz_{a,w,N}$, as the other terms are nuisance. And while the terms in the second line are asymptotically negligible (given exogeneity of regressors), the final term is not expected to be negligible. Therefore the denominator needs to be adjusted separately.

Although estimating time effects in the numerator would mechanically induce cross-sectional dependence, the time effects can be removed from the quadratic objects appearing in the denominator using a pooled estimate of the common time component in the $\ba_i$-objects.

 Any common time-effect vector admits the decomposition
\begin{equation}
\bd=c\biota_T+\bB'\btheta_d,
\qquad
c:=T^{-1}\biota_T'\bd,
\qquad
\btheta_d:=\bB\bd.
\end{equation}
We have $\ba_i=\bL_i\btheta_d + \bB\widetilde{\bM}_i\bu_i$, where $\bL_i=\bB\widetilde{\bM}_i\bB'$. In the more general case with additional common deterministic regressors, decompose $\bd=\bS\bgamma + \bB'\btheta_d$ where $\bS$ contains the deterministic regressors removed by $\widetilde{\bM}_i$.

The nuisance parameter $\btheta_d$ can therefore be estimated from the pooled weighted (or, alternatively, ordinary) least-squares regression\footnote{In the benchmark case in which $\widetilde\bM_i\bB'=\bB'$ for all $i$, so that $\bL_i=\bI$, the normalization $N^{-1}\sum_{i=1}^N w_i^2=1$ gives
$
\widehat\btheta_d
=
\frac{1}{N}\sum_{i=1}^N w_i^2\ba_i.
$ Thus, in this case, $\widehat\btheta_d$ is simply the weighted cross-sectional average of the time contrasts $\ba_i$.} of $\ba_i$ on $\bL_i$, with weights $w_i^2$:
\begin{equation}\label{eq:pooled}
\widehat\btheta_d
=
\left(\sum_{i=1}^N w_i^2\bL_i'\bL_i\right)^{-1}
\left(\sum_{i=1}^N w_i^2\bL_i'\ba_i\right)=\left(\sum_{i=1}^N w_i^2\bL_i\right)^{-1}
\left(\sum_{i=1}^N w_i^2\bL_i\ba_i\right),
\end{equation}
where the last equality uses symmetry and idempotency of $\bL_i$.

Define the residual after subtracting the estimated common time component as
\begin{equation}
\ba_i^\circ
:=
w_i\left(\ba_i-\bL_i\widehat\btheta_d\right)
=
w_i\bB\widetilde\bM_i\bu_i
-
w_i\bL_i\left(\widehat\btheta_d-\btheta_d\right).
\end{equation}
The denominator adjustment therefore removes from each weighted $\ba_i$ the estimated common time component. We make the following sufficient assumptions that ensure $\sqrt{N}$-consistency of $\hat{\btheta}_d$.

\begin{assumption}
\label{ass:pooled_time_effect}
Let
\(
\bX = (\bX_1,\ldots,\bX_N).
\)
The following conditions hold as $N\to\infty$:

\begin{enumerate}

\item[(i)] The disturbances satisfy
\(
\E[\bu_i\mid\bX]=\vzeros_T,
\)
are conditionally independent across $i$ given $\bX$, and
\(
\max_{1\leq i\leq N}
\E[\|\bu_i\|^2\mid\bX]
=
O_p(1).
\)

\item[(ii)] The weights are functions only of $\bX$ and satisfy
\(
\frac{1}{N}\sum_{i=1}^N w_i^4=O_p(1).
\)

\item[(iii)] For some positive definite matrix $\bG_L$,
\[
\frac{1}{N}\sum_{i=1}^N w_i^2\bL_i
\to_p
\bG_L,
\qquad
\lambda_{\min}(\bG_L)>0.
\]
\end{enumerate}
\end{assumption}
\begin{proposition}
Under Assumption~\ref{ass:pooled_time_effect},
\begin{equation}
\widehat\btheta_d-\btheta_d
=
O_p\left(N^{-1/2}\right),
\end{equation}
such that
\begin{align}
\frac{1}{N}\sum_{i=1}^N
\ba_i^\circ\ba_i^{\circ\prime}
&=
\frac{1}{N}\sum_{i=1}^N
w_i^2\bB\widetilde\bM_i\bu_i\bu_i'
\widetilde\bM_i\bB'
+
O_p\left(N^{-1}\right).
\end{align}
\end{proposition}
Thus we consistently estimate the covariance quantities associated with the nuisance-free weighted aggregate. The proof of this proposition is standard and thus omitted.

\section{Monte Carlo study}
\label{sec:monte_carlo}

We study the finite-sample behavior of the exact \(a\)-side weighted conditional test in a panel model with latent peer effects and an optional common time component.
\subsection{Data generating process}
For each replication, we generate a balanced panel with \(N\) units and \(T\) periods according to
\begin{equation}
\by_t
=
(\bI_N-\tau\bW)^{-1}
\bigl(
\balpha+\beta_x \bx_t+\gamma_{\mathrm{tr}}\delta_t\biota_N+\be_t
\bigr),
\qquad t=1,\dots,T,
\end{equation}
where \(\by_t=(y_{1t},\dots,y_{Nt})'\), \(\bx_t=(x_{1t},\dots,x_{Nt})'\), \(\balpha=(\alpha_1,\dots,\alpha_N)'\), and \(\be_t=(e_{1t},\dots,e_{Nt})'\).

For each unit \(i\), the regressor process is generated as $
x_{i1}
=
\mu_x + q_i + \frac{\xi_{i1}}{\sqrt{1-\rho_x^2}},
$
and, for \(t=2,\dots,T\),
$
x_{it}
=
\mu_x+\rho_x(x_{i,t-1}-\mu_x)+0.3\,q_i+\xi_{it},
$
where \(q_i \sim \mathcal{N}(0,1)\) and \(\xi_{it}\sim \mathcal{N}(0,1)\), independently across \(i\) and \(t\). The unit-specific intercepts satisfy \(\alpha_i \sim \mathcal{N}(0,\sigma_\alpha^2)\), \(i=1,\dots,N\). Finally, the error terms follow an AR(1) process within unit $
e_{i1}\sim \mathcal{N}(0,1),$
and, for \(t=2,\dots,T\),
$
e_{it}
=
\rho_e e_{i,t-1}
+
\sqrt{1-\rho_e^2}\,\varepsilon_{it}$, where $
\varepsilon_{it}\sim \mathcal{N}(0,1),
$
independently across \(i\) and \(t\).

The common time component is based on the centered trend
\begin{equation}
\delta_{t}^* = t-\frac{T+1}{2},
\qquad
\delta_t = \frac{\delta_{t}^*}{
\sqrt{\frac{1}{T}\sum_{s=1}^T (\delta_{s}^{*})^2}.}
\end{equation}
Hence \(\gamma_{\mathrm{tr}}=0\) corresponds to the case without a common time effect, while \(\gamma_{\mathrm{tr}}\neq 0\) introduces a common time component.

The scalar \(\tau\ge 0\) governs the strength of peer effects. So that under the null, \(\tau=0\),
\begin{equation}
\by_t=\balpha+\beta_x \bx_t+\gamma_{\mathrm{tr}}\delta_t\biota_N+\be_t.
\end{equation}

For the power analysis $\tau>0$, we need to specify the latent weighting matrix $\bW$.
For each unit \(i\), define $
c_i
=
z(\widehat{\mu}_{x}(i))
+
z(\widehat{\gamma}^{(1)}_{x}(i))
+
0.5\,z(\widehat{\sigma^{2}}_{x}(i)),
$
where $z(\cdot)$ indicates cross-sectional standardization to mean zero and variance one, and
\begin{equation}
\widehat{\mu}_{x}(i)=\frac{1}{T}\sum_{t=1}^T x_{it},\quad
\widehat{\sigma^{2}}_{x}(i)=\frac{1}{T-1}\sum_{t=1}^T (x_{it}-\bar x_i)^2,\quad
\widehat{\gamma}^{(1)}_{x}(i)
=
\frac{1}{T-1}\sum_{t=2}^T (x_{it}-\bar x_i)(x_{i,t-1}-\bar x_i).
\end{equation}
The units are partitioned into three latent groups according to terciles of \(c_i\). Let \(g_i\in\{1,2,3\}\) denote the group assignment. Conditional on \((g_i,g_j)\), the adjacency matrix \(\bA=(A_{ij})\) is generated as
\begin{equation}
A_{ij}=A_{ji}=1\{U_{ij}\le p_{ij}\},
\qquad
A_{ii}=0,
\end{equation}
where \(U_{ij}\sim U[0,1]\), and
\begin{equation}
p_{ij}
=
\begin{cases}
p_{\mathrm{in}}, & g_i=g_j,\\
p_{\mathrm{out}}, & g_i\neq g_j.
\end{cases}
\end{equation}
In the dense design, \(p_{\mathrm{in}}=0.7\) and \(p_{\mathrm{out}}=0.15\), while in the sparse design, \(p_{\mathrm{in}}=0.25\) and \(p_{\mathrm{out}}=0.02\). For the sparse design, if unit \(i\) has more than \(2k\) neighbors, only the \(2k\) neighbors with \(c_j\) closest to \(c_i\) are retained. Finally, the network matrix is scaled by its spectral radius:
\begin{equation}
\bW = \frac{\bA}{\rho(\bA)}.
\end{equation}
Finally, to select the weights, we choose $\bs$ from \eqref{eq:s} as
\begin{equation}
s_i=c_i
=
z(\widehat{\mu}_{x}(i))
+
z(\widehat{\gamma}^{(1)}_{x}(i))
+
0.5\,z(\widehat{\sigma^{2}}_{x}(i)),
\qquad
\bs=(s_1,\dots,s_N)'.
\end{equation}

Unless otherwise stated, the parameter values are \(\beta_x = 1\), \(\mu_x = 1\), \(\rho_x = 0.6\), \(\rho_e = 0.4\), \(\sigma_\alpha = 1\), and \(\gamma_{\mathrm{tr}} = 1\). The sparse-neighborhood trimming parameter is \(k=2\). The Monte Carlo design uses \(1000\) draws of the DGP, and critical values are calculated over \(2000\) draws. The rejection probability is evaluated over \(\tau \in \{0,0.2,0.4,0.6,0.8\}\). The panel dimensions considered are $N=250$ and $T=\{4,8\}$. Unreported results for $N=100$ show qualitatively the same patterns.

\begin{figure}[t]
    \caption{Simulation results: Power ($N=250$, $T=4$).}
    \label{fig:feasible_no_trend}
    \centering
    \begin{subfigure}{0.48\textwidth}
    \caption{No trend, dense}
    \includegraphics[width=\textwidth]{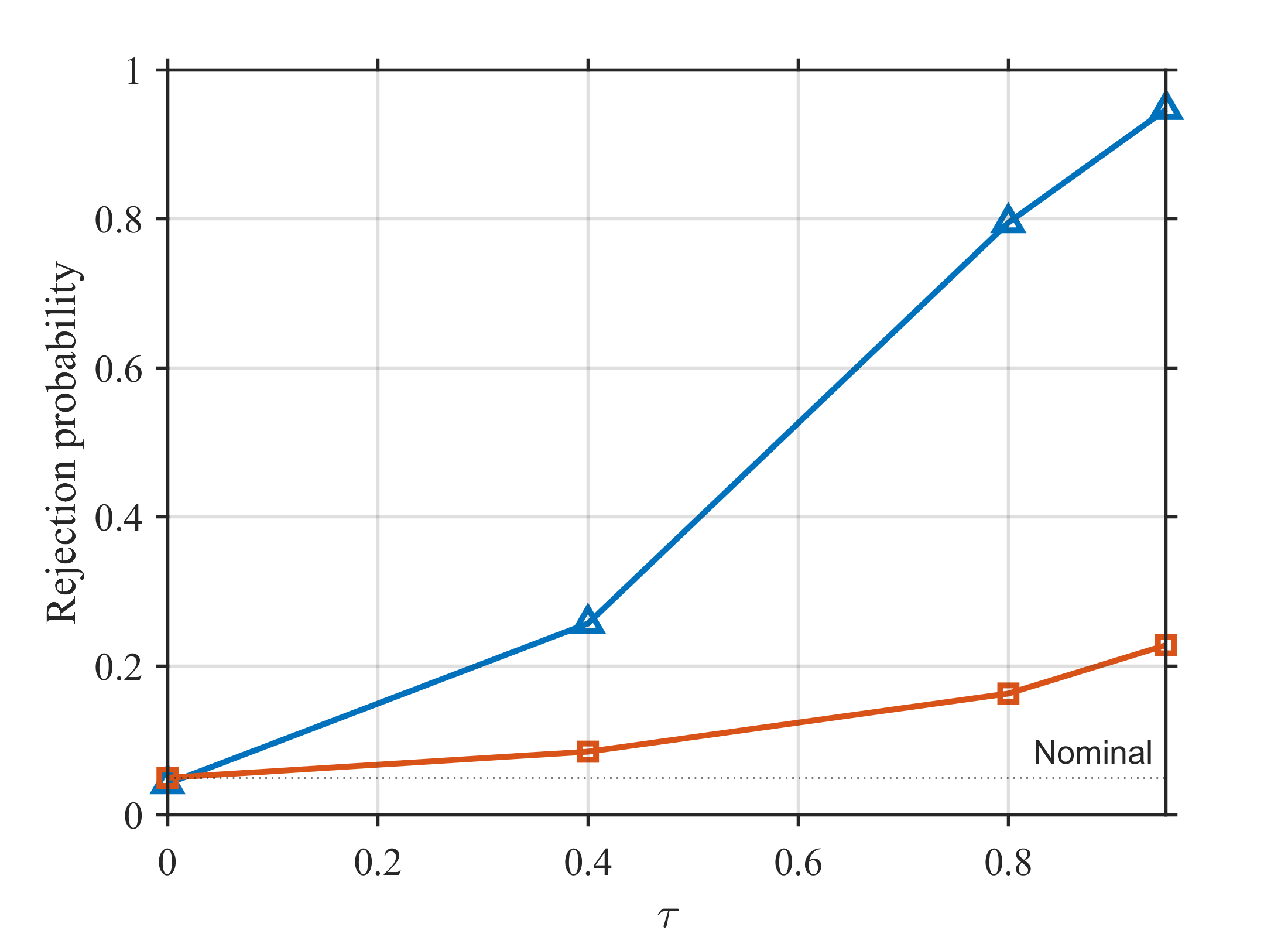}
    \end{subfigure}
    \begin{subfigure}{0.48\textwidth}
        \caption{No trend, sparse}
     \includegraphics[width=\textwidth]{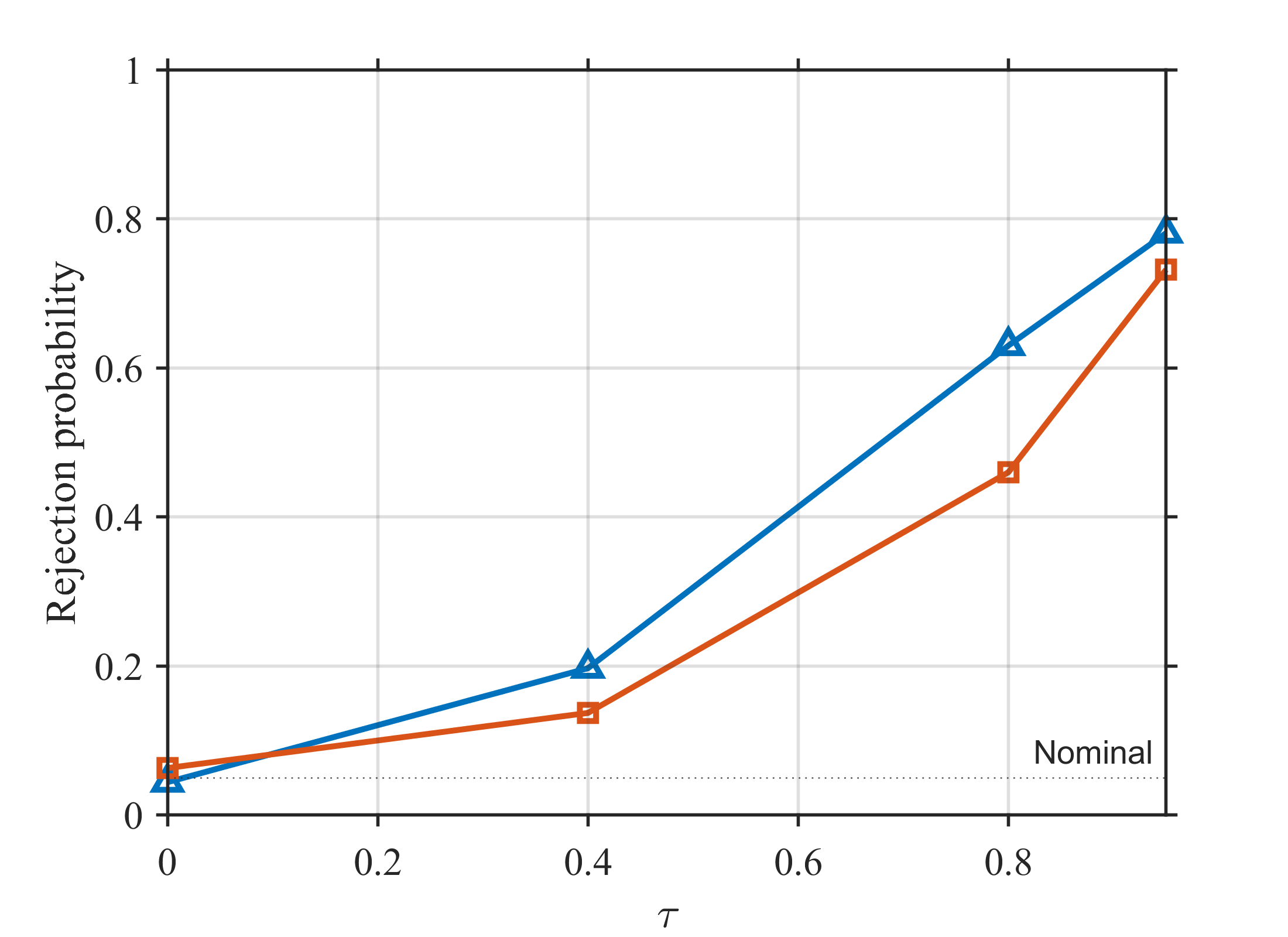}
     \end{subfigure}
        \begin{subfigure}{0.48\textwidth}
    \caption{Trend, dense}
    \includegraphics[width=\textwidth]{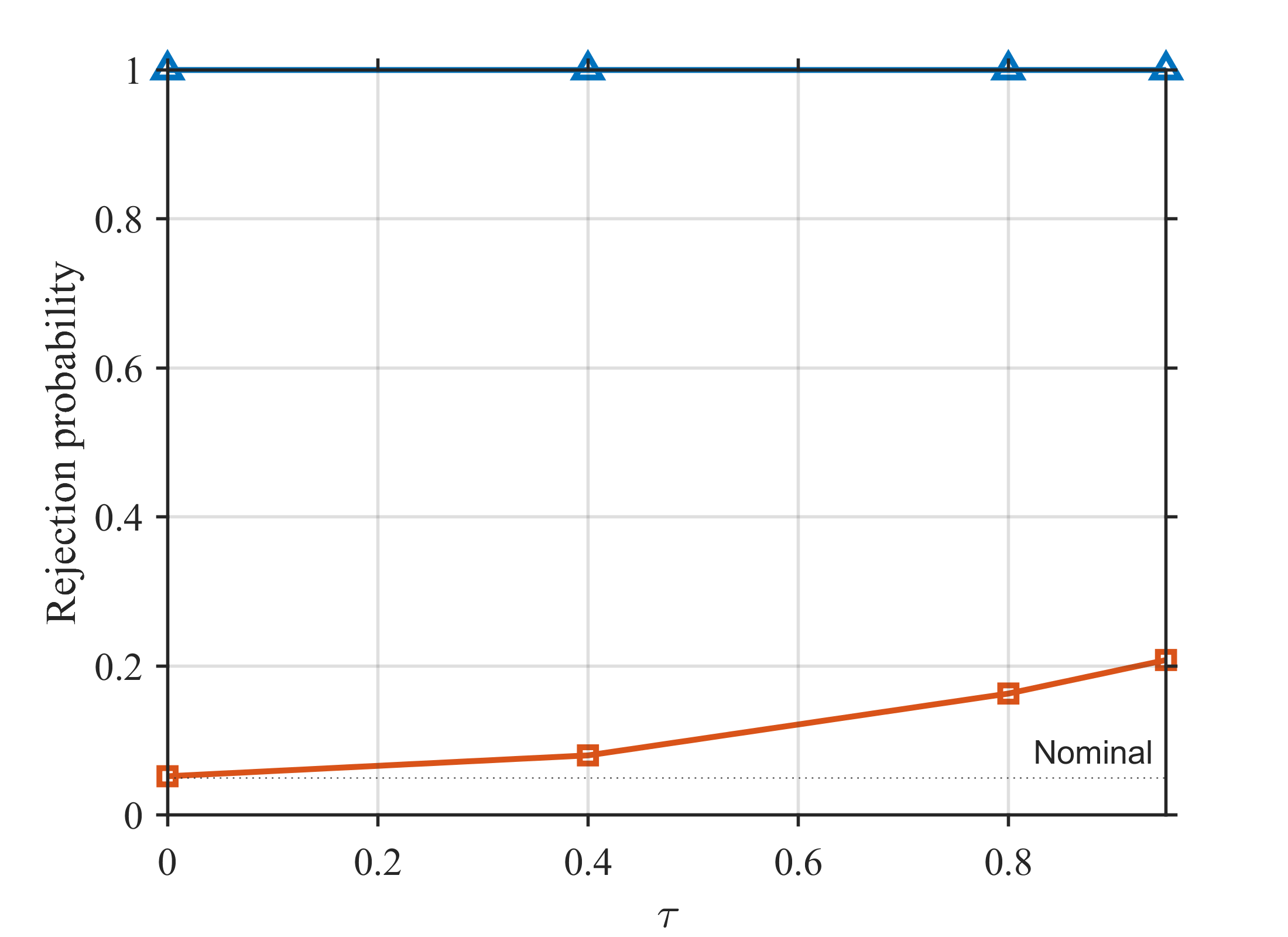}
    \end{subfigure}
    \begin{subfigure}{0.48\textwidth}
        \caption{Trend, sparse}
     \includegraphics[width=\textwidth]{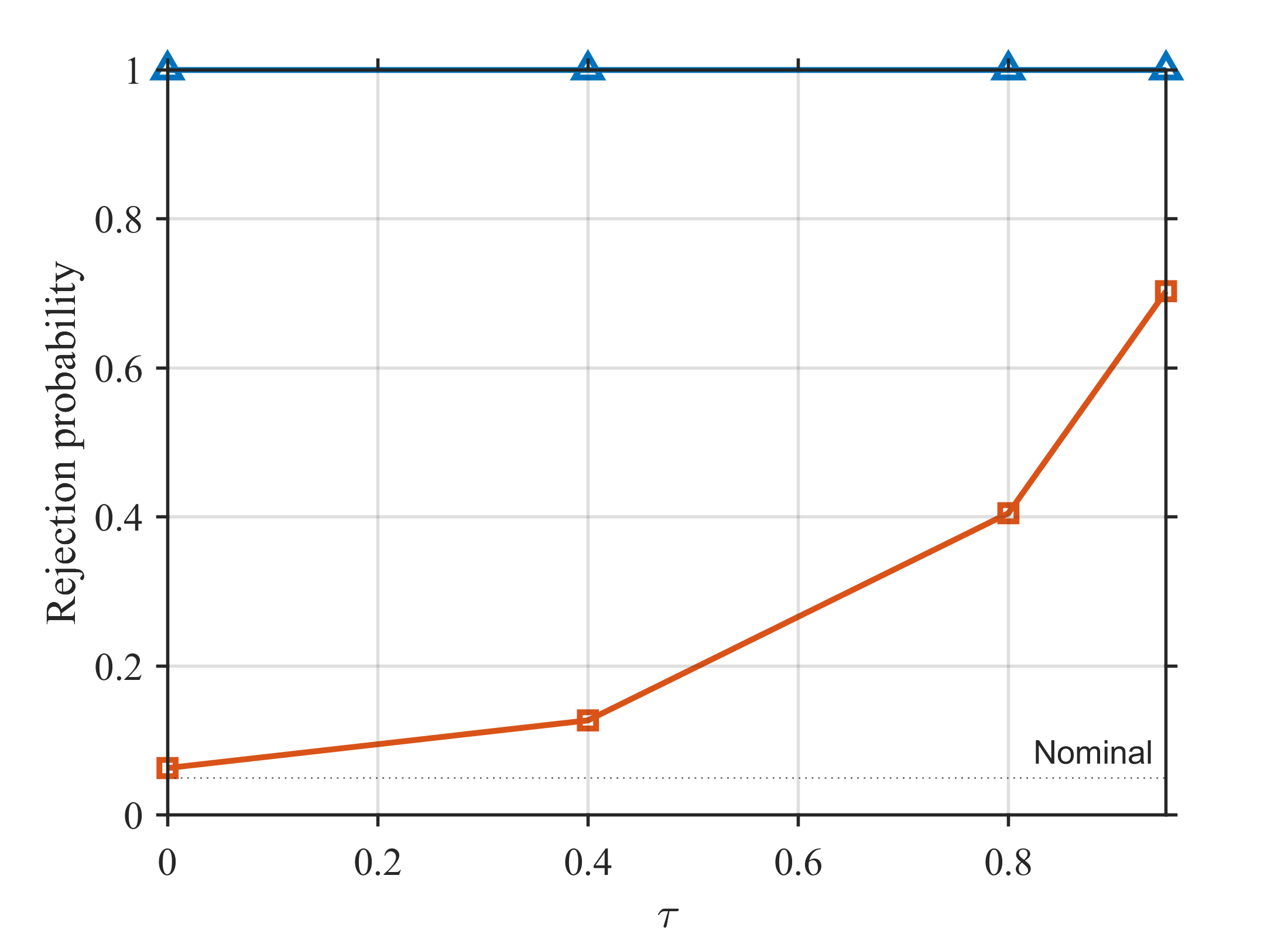}
     \end{subfigure}
     \begin{minipage}{0.95\textwidth}
    \small \textit{Note}: Upper panel show power under the linear peer effects model without a trend under both a dense (left) and sparse (right) network. The lower panel shows the same results for a model with a trend. Blue lines indicate results from the unweighted statistic with $b_{\star} = \max(\text{tr}(\bGamma^2),0)$, orange lines show the weighted version. In this simulation design the choice of $b_{\star}$ does not affect the results.
     \end{minipage}
\end{figure}

\subsection{Results}
The results are reported in Figure~\ref{fig:feasible_no_trend} for $T=4$ and in Figure~\ref{fig:feasible_no_trend2} for $T=8$. The figures in the left panel correspond to the dense network design and the figures in the right panel to the sparse network design.

\begin{figure}[t]
    \caption{Simulation results: Power ($N=250$, $T=8$).}
    \label{fig:feasible_no_trend2}
    \centering
    \begin{subfigure}{0.48\textwidth}
    \caption{No trend, dense}
    \includegraphics[width=\textwidth]{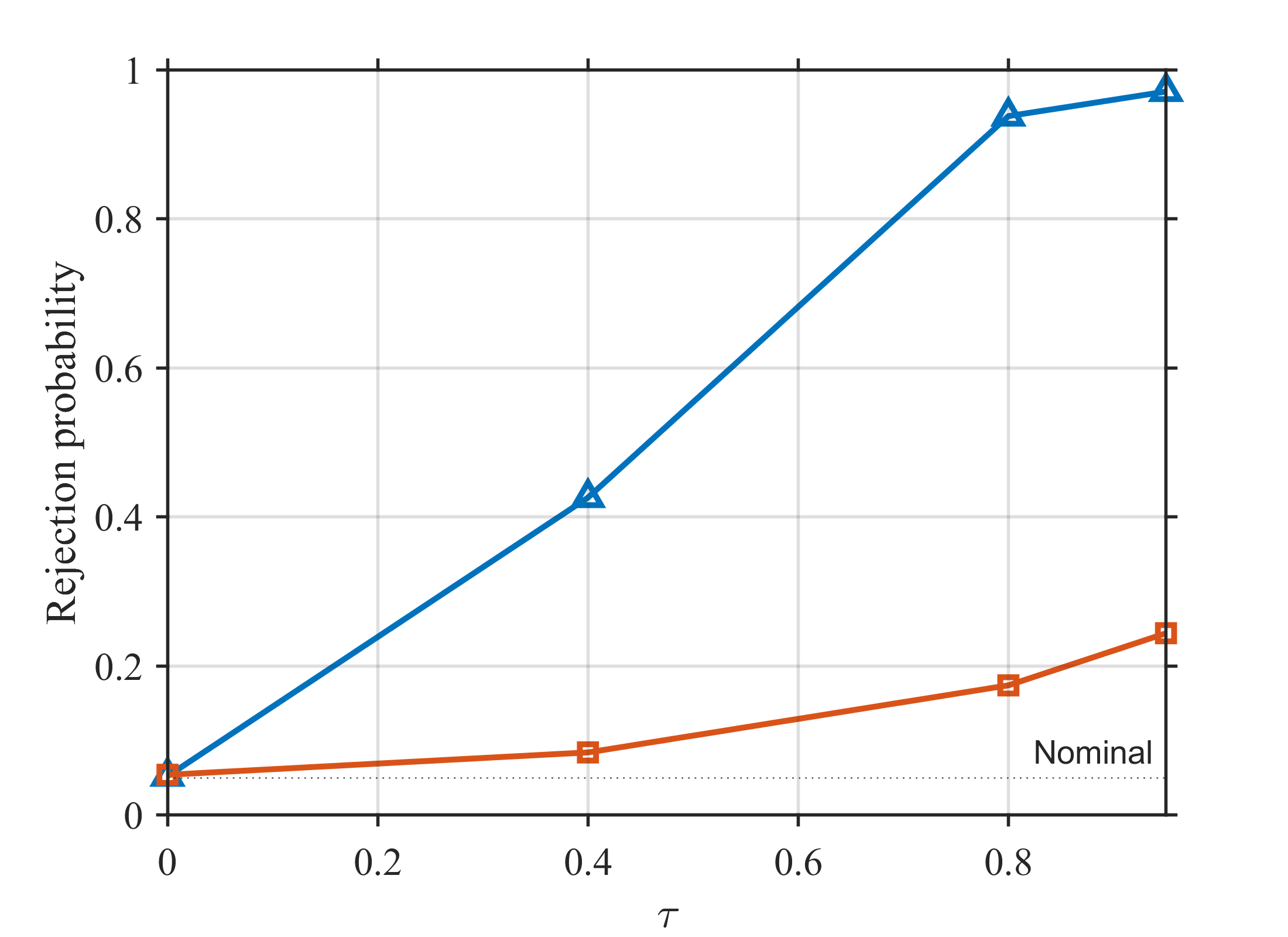}
    \end{subfigure}
    \begin{subfigure}{0.48\textwidth}
        \caption{No trend, sparse}
     \includegraphics[width=\textwidth]{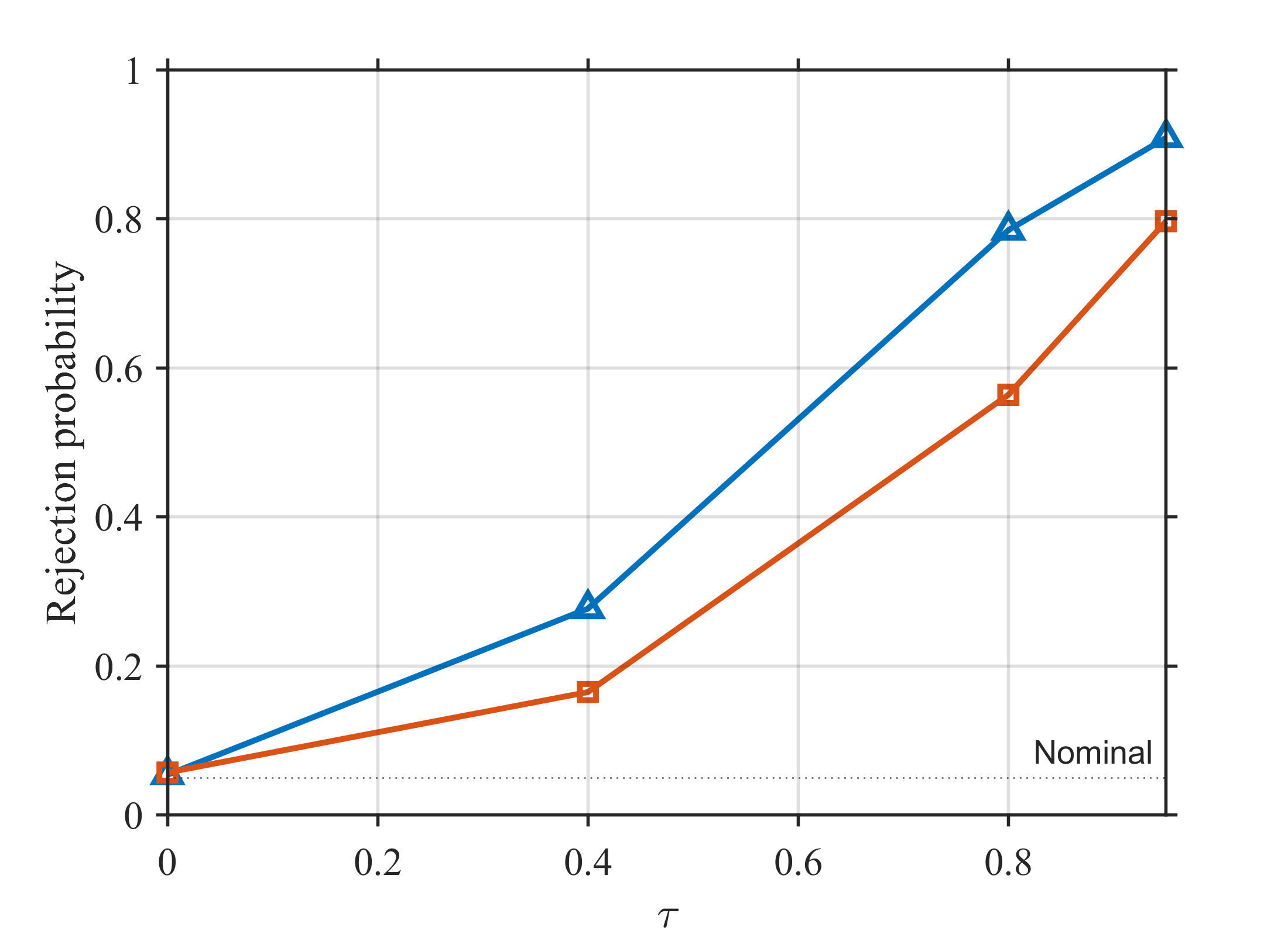}
     \end{subfigure}
        \begin{subfigure}{0.48\textwidth}
    \caption{Trend, dense}
    \includegraphics[width=\textwidth]{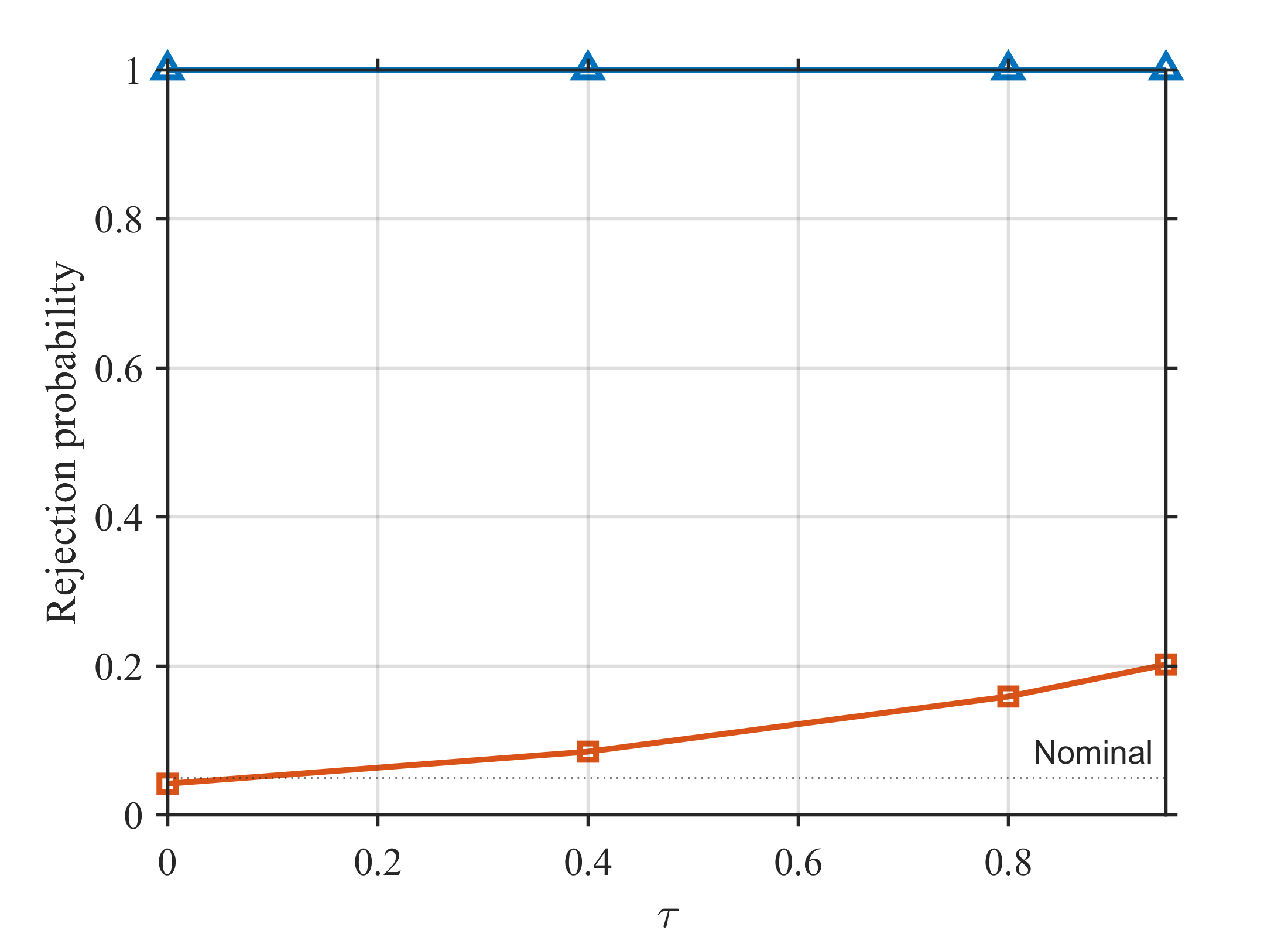}
    \end{subfigure}
    \begin{subfigure}{0.48\textwidth}
        \caption{Trend, sparse}
     \includegraphics[width=\textwidth]{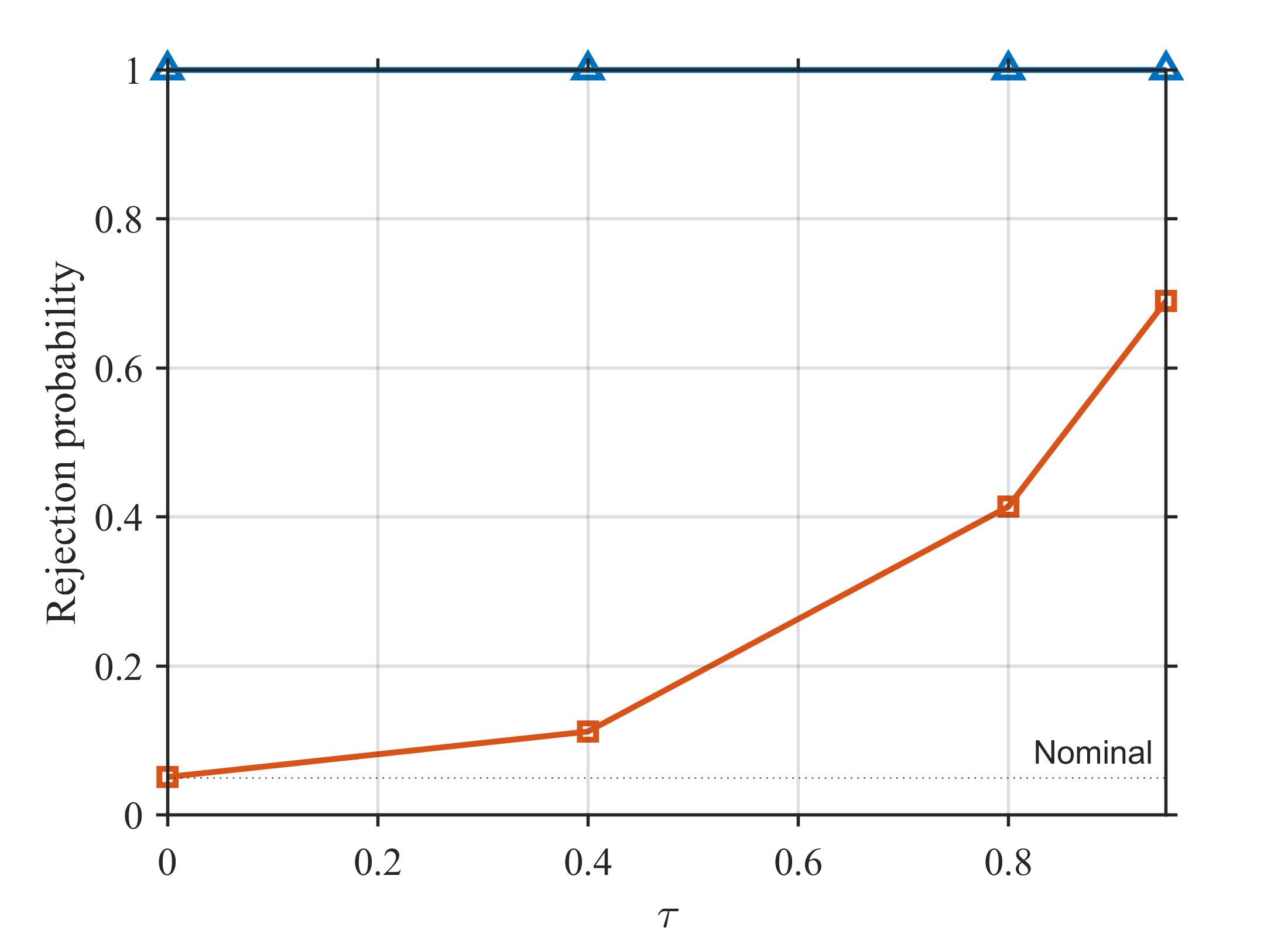}
     \end{subfigure}
     \begin{minipage}{0.95\textwidth}
    \small \textit{Note}: Upper panel show power under the linear peer effects model without a trend under both a dense (left) and sparse (right) network. The lower panel shows the same results for a model with a trend. Blue lines indicate results from the unweighted statistic with $b_{\star} = \max(\text{tr}(\bGamma^2),0)$, orange lines show the weighted version. In this simulation design the choice of $b_{\star}$ does not affect the results.
     \end{minipage}
\end{figure}
When no common time component is present, both the unweighted conditional test and the weighted conditional test have rejection probabilities close to the nominal level under \(\tau=0\). This confirms that the weighted construction does not introduce a size distortion when the time-effect nuisance is absent. The unweighted statistic is more powerful in this case, especially in the dense design. This is expected: without a common time component, the exact balancing restriction is unnecessary and the weighted statistic pays a power cost for imposing it.

The picture changes sharply when the common time component is present. The unweighted conditional test rejects with probability close to one even at \(\tau=0\), showing that the common time component is interpreted as cross-sectional dependence. By contrast, the weighted test remains close to the nominal level under the null. This confirms that the balancing condition removes the time-effect contribution from the aggregate score. As \(\tau\) increases, the power curve is very close to the curve in the model without a trend.

Overall, the simulations illustrate the tradeoff created by the weighted procedure. In designs without time effects, weighting is conservative relative to the unweighted statistic and can reduce power. In designs with common time effects, however, weighting is essential for size control and restores a meaningful interpretation of rejection as evidence against the null of no peer effects.

\section{Empirical illustration}
\label{sec:application_relative_wages}
\subsection{Peer (spillover) effects in relative wages}
We apply our procedure to investigate presence of industry-level spillovers between in the US manufacturing industries. For this purpose we use the data from
\citet{voigtlander2014skill} that was used by \citet{yin2021focused} and \citet{Juodis2022CCER}, among others.
The data contain \(N=313\) cross-sectional sectors observed annually from 1958 to
2005. We use the narrow specification from \citet{yin2021focused}.
\begin{equation}
\ln\left(\frac{w_{L.it}}{w_{H.it}}\right)
=
\beta_{1i}\ln(\sigma_{it})
+
\beta_{2i}\ln\left(\frac{H_{it}}{L_{it}}\right)
+
\blambda_i'\blf_{t} + \tau_t
+
\varepsilon_{it}.
\end{equation}

\begin{figure}[t]
\centering
\caption{Application: conditional $p$-values without incidental trends}
\label{fig:relative_wage_notrend}
\includegraphics[width=\textwidth,trim=2cm 0 2cm 0,clip]{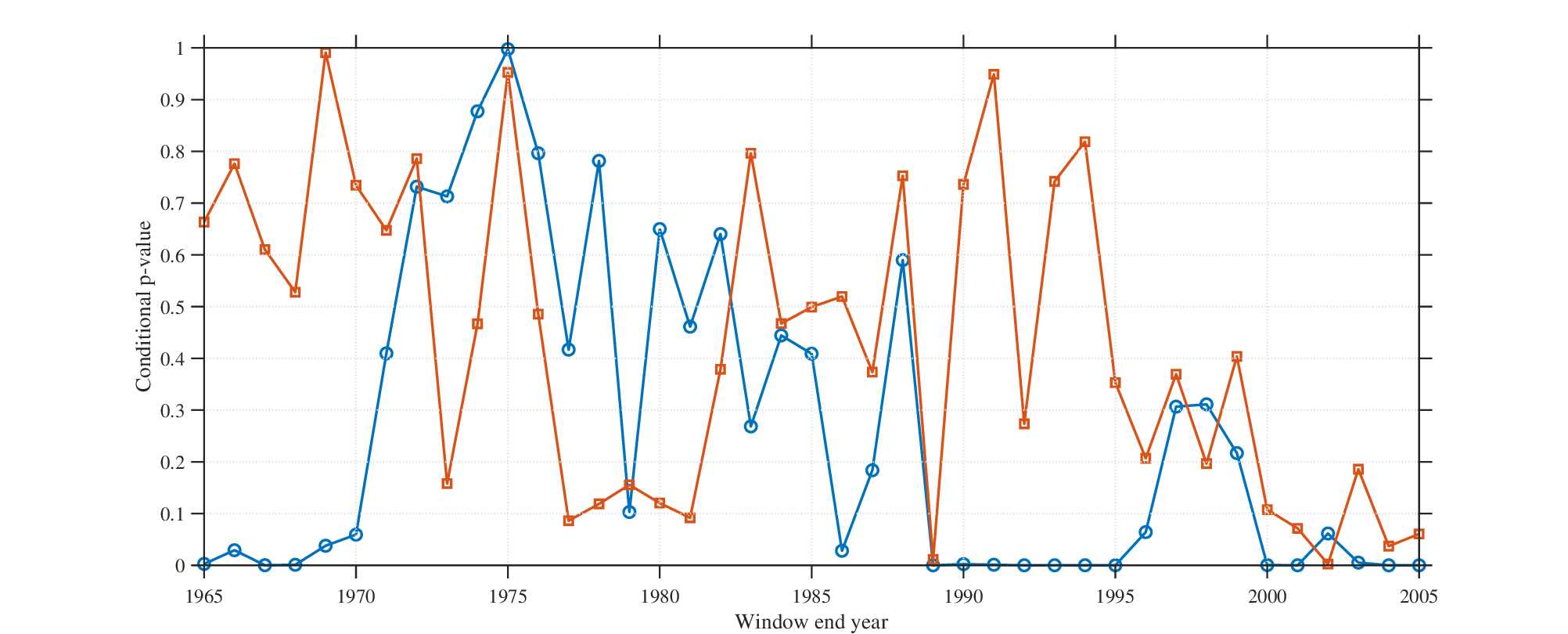}
\begin{minipage}{0.95\textwidth}
\small \textit{Note}: Red corresponds to the weighted statistic; Blue corresponds to the unweighed statistic.
\end{minipage}
\end{figure}

The outcome is the log of the ratio of wage for low skilled workers and high-skilled workers, \(\sigma_{it}\) is the input intensity skill measure and $H_{it}$ and $L_{it}$ are the number of high-skilled and low-skilled workers respectively. Based on some visual evidence in \citet{Juodis2022CCER},\footnote{See the corresponding Supplementary Online Appendix.} it is reasonable to expect that the largest part of the common factor can be approximated by linear deterministic trend, so we consider a variant of the model where $\blf_t = (1,t)'$. Similarly to our Monte Carlo study, we consider test statistics with and without the time effects $\tau_t$.

We construct our test statistic using rolling windows
\(\mathcal T_r=\{r,\ldots,r+T_{w}-1\}\) for some pre-specified length \(T_{w}\). Let \(\by_{i,r}\) and \(\bX_{i,r}\) denote
the corresponding \(T_{w}\times1\) outcome vector and an \(T_{w}\times2\) (narrow) regressor
matrix.
The projected component is
\begin{equation}
\bp_{i,r}:=\bP_{i,r}\by_{i,r},
\qquad
\bP_{i,r}
:=
\bX_{i,r}(\bX_{i,r}'\bX_{i,r})^{-1}\bX_{i,r}'.
\end{equation}
Throughout this application the projection side uses the two narrow-model regressors directly and does not orthogonalize with respect to the time trend. We set \(T_{w}=8\).

\begin{figure}[t]
\centering
\caption{Application: conditional $p$-values with incidental trends}
\label{fig:relative_wage_trend}
\includegraphics[width=\textwidth,trim=2cm 0 2cm 0,clip]{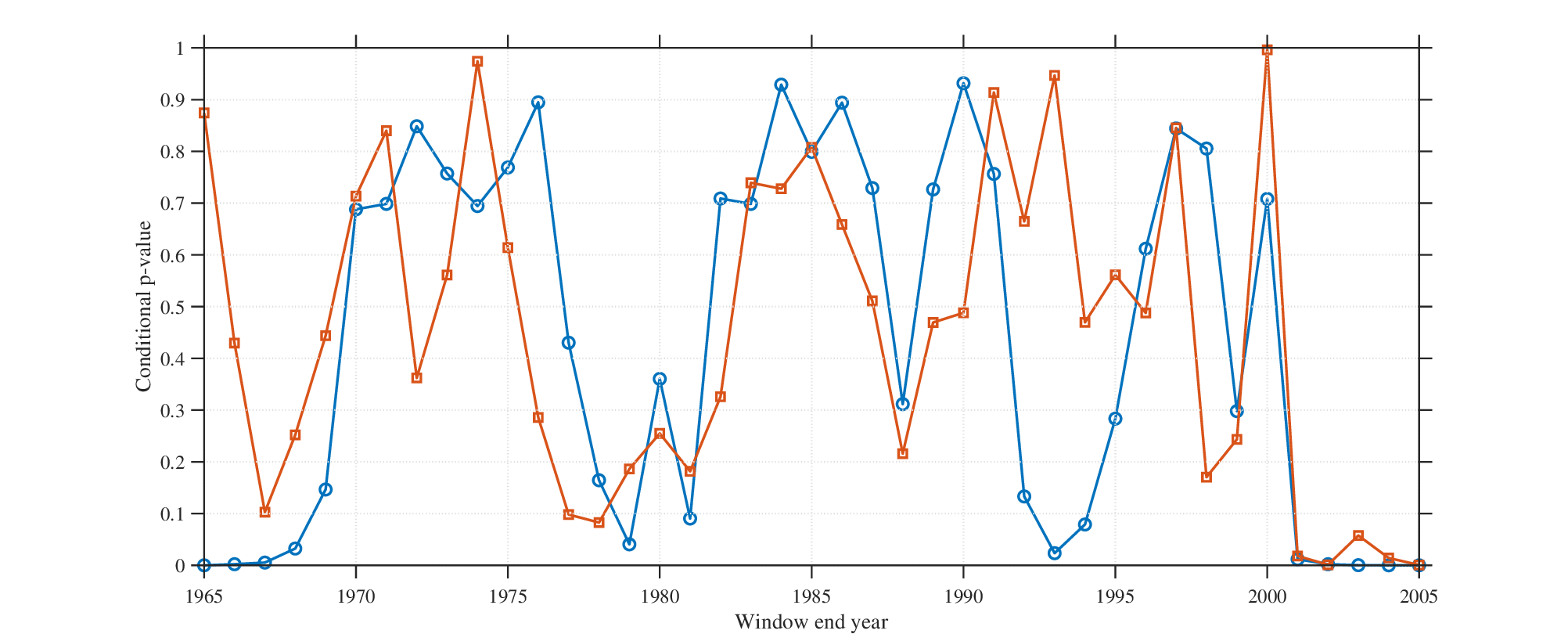}
\begin{minipage}{0.95\textwidth}
\small \textit{Note}: Red corresponds to the weighted statistic; Blue corresponds to the unweighed statistic.
\end{minipage}
\end{figure}
Figures~\ref{fig:relative_wage_notrend} and
\ref{fig:relative_wage_trend} report rolling conditional \(p\)-values, indexed
by the final year of the eight-year window. The no-trend specification in
Figure~\ref{fig:relative_wage_notrend} displays several episodes of rejection.
The unweighted statistic is especially small in windows ending in the late
1960s, around the early 1990s, and again at the end of the sample. The weighted
statistic mainly shows lower $p$-values after 2000. \

Figure~\ref{fig:relative_wage_trend} shows that adding individual linear trends
changes the interpretation. First, it brings the weighted and the unweighted statistic much closer, suggesting that after accounting for incidental trends, the common time effects have much less of an effect. What remains is the rejection of the null in the post 2000 years. However, because the rolling
windows overlap, individual \(p\)-values should not be interpreted as independent
tests.

\begin{figure}[t]
\centering
\caption{Application: group pair contributions, window 1998-2005}
\label{fig:heat_map}
\includegraphics[width=\textwidth,trim=0 0 5.25cm 0,clip]{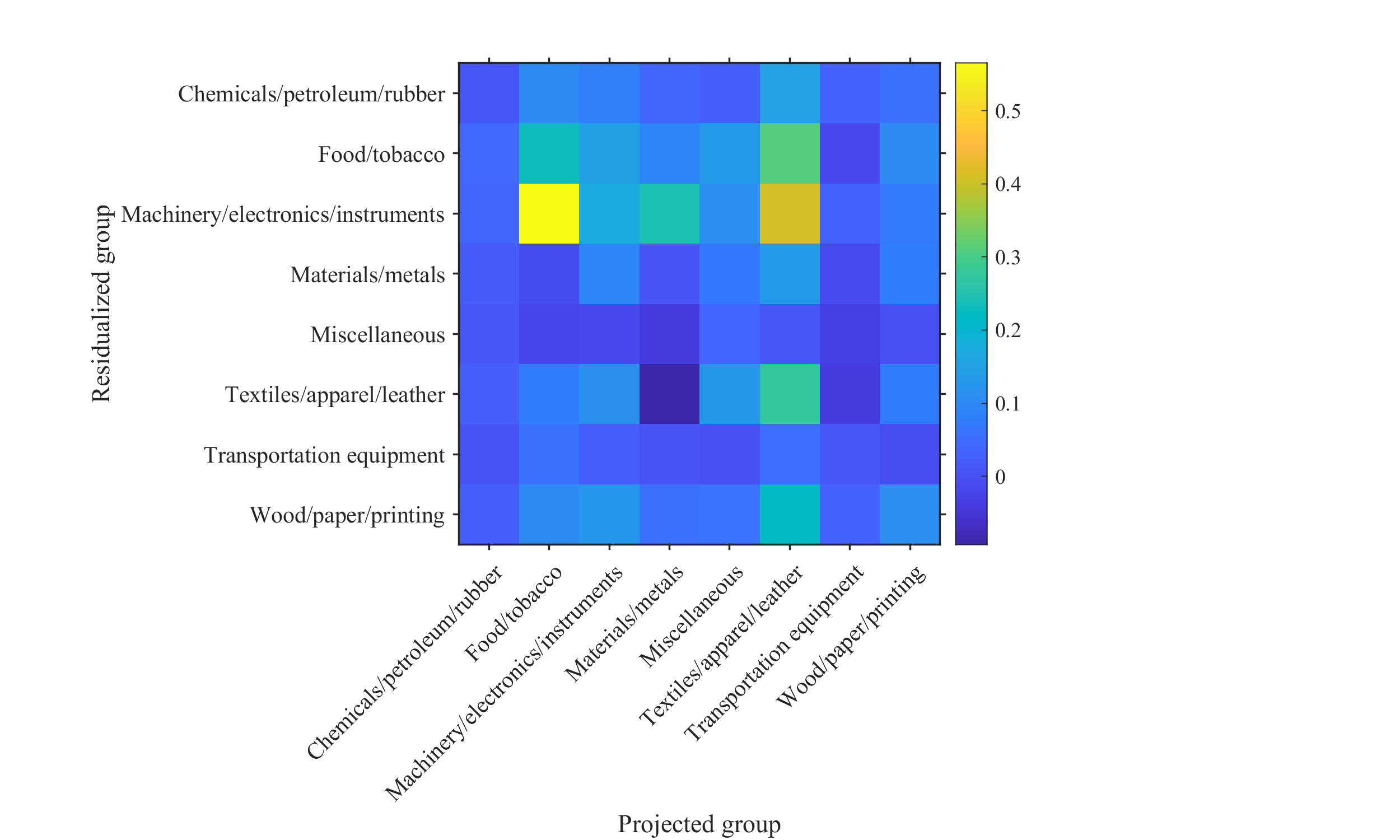}
\end{figure}

\subsection{Potential origins of dependence}
To diagnose the source of the late-sample rejections, we decompose the numerator
of the unweighted statistic by broad two-digit SIC manufacturing groups. Let
\(G_i\) denote the industry group of sector \(i\). For groups \(g\) and \(h\),
define
\begin{equation}
\widehat S_{g,h}
:=
\frac{1}{N}
\left(\sum_{j:G_j=h}\bp_j\right)'
\left(\sum_{i:G_i=g}\ba_i\right),
\end{equation}
so that
\(
\widehat S=\sum_{g,h}\widehat S_{g,h}.
\)
Figure~\ref{fig:heat_map} shows a heatmap where rows correspond to the residualized \(a\)-side group \(g\), and
columns correspond to the projected \(p\)-side group \(h\). Thus the
decomposition is directional in the sense of the statistic: a large cell
\((g,h)\) means that the projected outcome component of group \(h\) is aligned
with the residualized outcome component of group \(g\). In the linear peer effects model, this is consistent with outcomes from group \(h\) appearing in the outcome equations for group \(g\).

The decomposition shows that the post-2000 rejection is driven mainly by an
interaction between food/tobacco and
machinery/electronics/instruments. In the 1998--2005 window, the largest
contribution comes from the projected food/tobacco component and the
residualized machinery/electronics/instruments component. Under a literal
linear peer-effects interpretation, this corresponds to the
relative wage gap in food/tobacco appearing as an explanatory component in the
relative-wage equation for machinery/electronics/instruments, but not
conversely. A possible explanation is that food/tobacco captures a broad,
domestic manufacturing wage-gap component that remains unexplained in
machinery/electronics/instruments after controlling for the narrow-model
covariates and individual trends. The reverse cell is small because
food/tobacco is not affected by machinery-specific shocks to the wage gap that can arise from offshoring,
technology adoption, and global supply-chain restructuring.

\section{Conclusion}
\label{sec:conclusion}

This paper proposes testing framework for peer effects in short panels when the network is unknown. Proposed construction avoids estimating the network by aggregating pairwise comparisons between covariate-projected outcomes of potential peers and residualized own outcomes. This yields a fixed-\(T\), large-\(N\) procedure for testing the global null of no peer effects.

The theory shows that standard Gaussian critical values are valid when the projected component is strong, but can fail when this component is small. We therefore develop a conditional inference procedure based on the Gaussian limit experiment. Conditioning removes the local nuisance parameter and leads to feasible simulated critical values. We also propose an extension of our method that allows additive time effects in the true model.

The simulations support the theory. As expected from our theory, the use of Gaussian critical values can be misleading if the projected component is small. The proposed conditional procedure provides more reliable size control and retains power against latent peer-effect alternatives. Empirical relevance of the proposed method is illustrated using the real dataset.

Finally, in this paper we limited our attention to the setting where $N$ is large, but $T$ is fixed. It is natural to ask what happens to our procedure when also $T$ increases. We conjecture that our procedure will deliver asymptotically valid inference even in such cases as long as the unweighted statistic is considered. For the weighted statistic, the analysis is more subtle due to the high-dimensional regression adjustment used in the construction of the denominator. A formal analysis of such extensions, however, is beyond the scope of this study.

\paragraph{Funding} Financial support from the Dutch Research Council (NWO) under research grants VI.Vidi.241E.008 (first author) and VI.Vidi.231E.030 (second author) is gratefully acknowledged.
\setlength{\bibsep}{0pt}
\bibliographystyle{chicago}
\bibliography{bibliography}
\appendix

\section{Additional simulation results}

\subsection{Power in the limiting experiment}\label{app:compareCondKlei}
Suppose \(\bGamma=\mu\bI_2\). Strictly speaking, this case does not fit the theory as $\text{tr}(\bGamma)=0$ in our specific set-up, but it serves as a useful example to study size and power properties of the proposed test. This case is rotationally invariant. Therefore the conditional law depends on \(\blm\) only through \(m=\|\blm\|\). Moreover, \(\operatorname{tr}(\bGamma^2)=\operatorname{tr}(\bGamma\bGamma')=2\mu^2\), so the three normalizations coincide.

At \(m=0\), \(\bh=\mu\bz_p\). Hence
\begin{equation}
R_\star(\bh_0)
=
\frac{\mu\|\bz_a\|^2}
{\sqrt{2\mu^2+\mu^2\|\bz_a\|^2}}
=
\operatorname{sgn}(\mu)
\frac{\|\bz_a\|^2}{\sqrt{2+\|\bz_a\|^2}}.
\end{equation}
Thus the conditional distribution at the origin is non-Gaussian and, for \(\mu>0\), supported on positive values. Figure~\ref{fig:scalariso_results} shows that Gaussian critical values overreject as \(\mu\) increases, while the conditional version offers close to nominal size.

\begin{figure}[t]
\caption{Scalar-isotropic \(2\times2\) design.}
\label{fig:scalariso_results}
\centering
\begin{subfigure}{0.48\textwidth}
    \caption{Rejection at \(\bh_0=\mathbf 0\).}
    \label{fig:scalariso_h0}
  \centering
    \includegraphics[width=\linewidth]{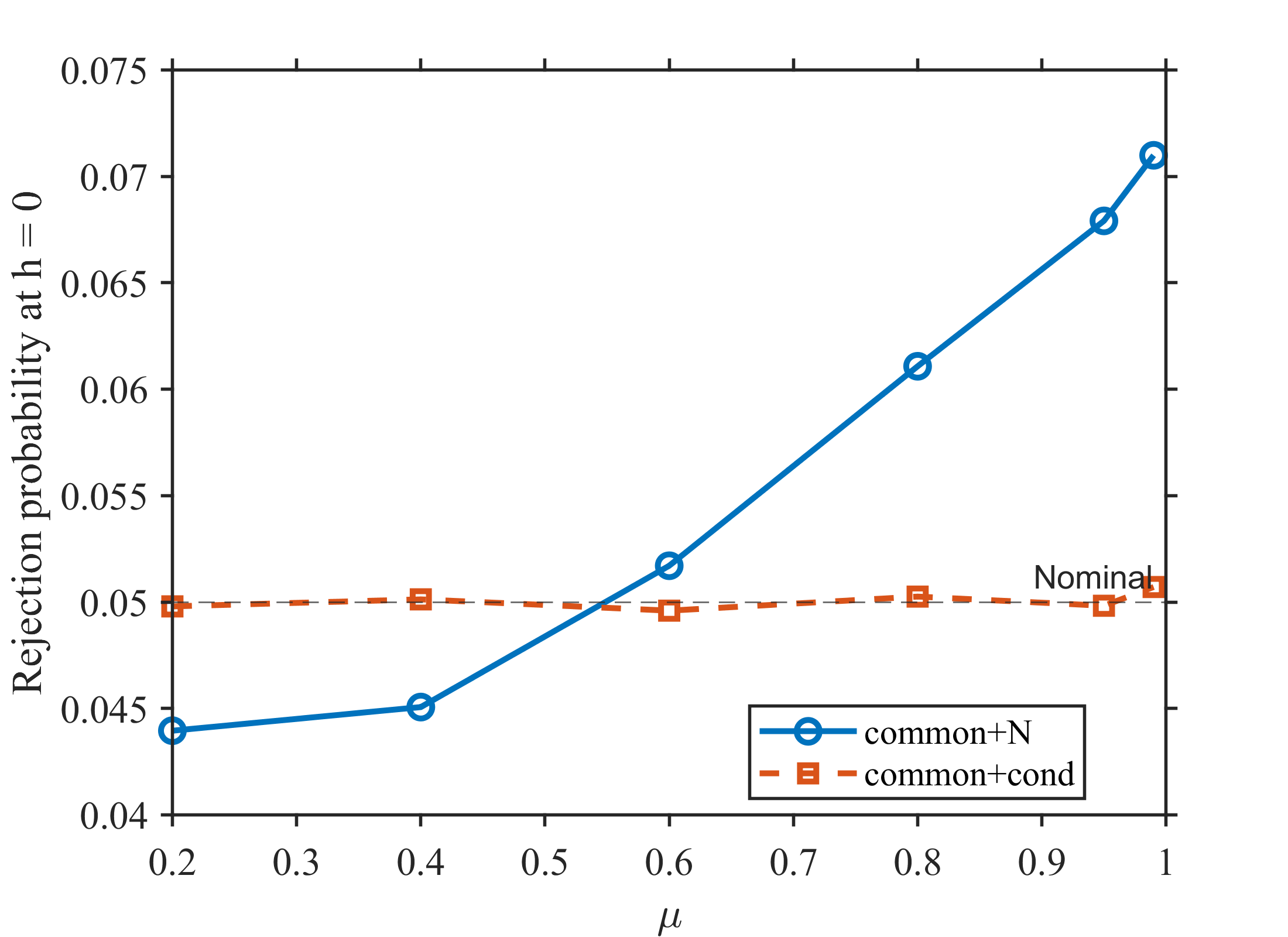}
\end{subfigure}
\hfill
\begin{subfigure}{0.48\textwidth}
      \caption{Worst-case rejection over the \(\bh_0\)-grid.}
    \label{fig:scalariso_worst}
 \centering
    \includegraphics[width=\linewidth]{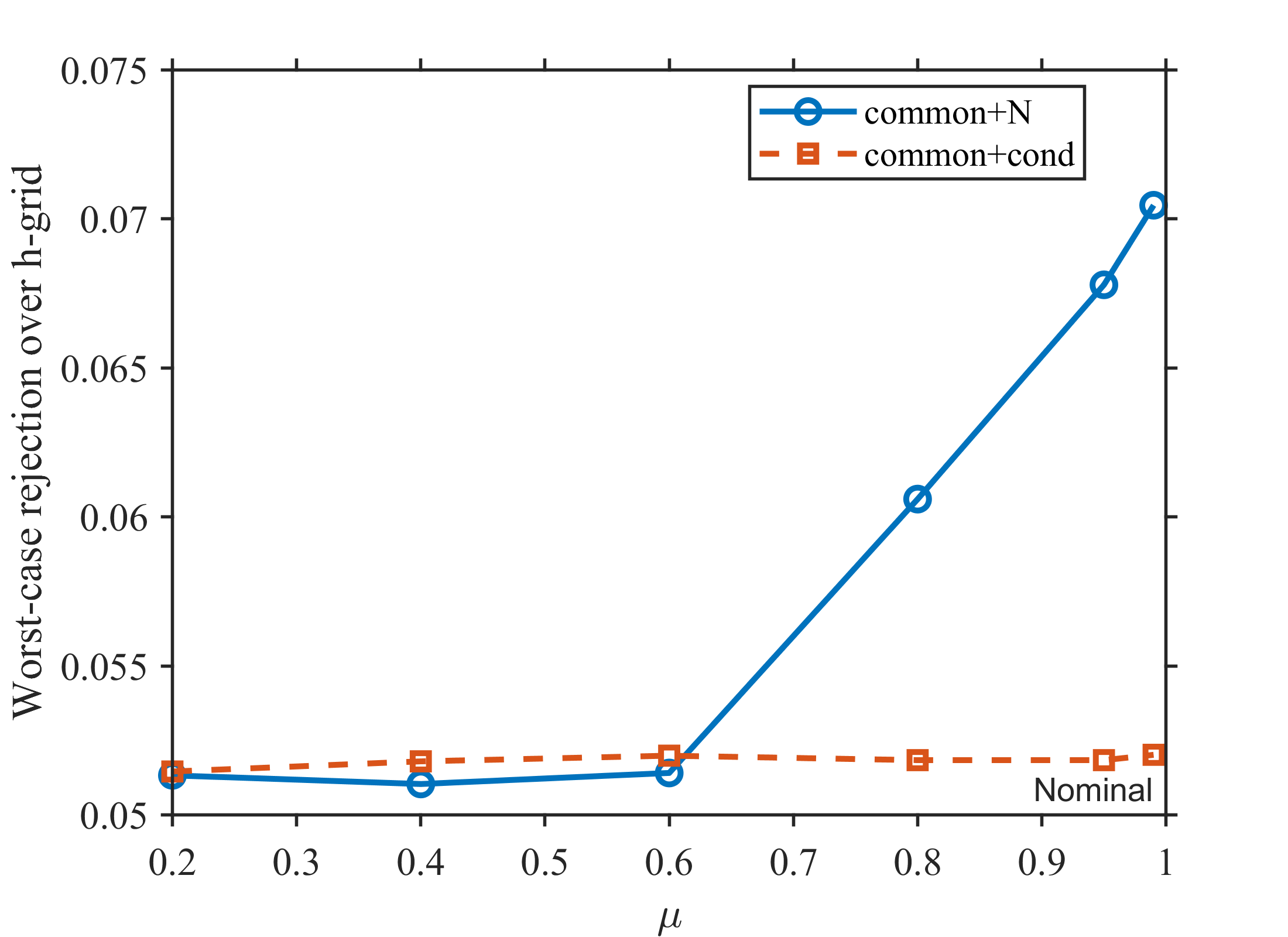}
\end{subfigure}
\begin{minipage}{0.92\textwidth}
\small \textit{Note}: The figure reports rejection probabilities in the scalar-isotropic design \(\bGamma=\mu\bI_2\). The conditional law is rotationally invariant and depends only on \(m=\|\blm\|\).
\end{minipage}
\end{figure}

This setting is also well-suited to study the power properties of the conditional and score test. The canonical experiment under the alternative is
\begin{equation}
\bz_a = \bmu_a + \bepsi,
\qquad
\bz_p = \bh + \bmu_p + \bGamma \bepsi + \bet,\qquad \bepsi\perp\bet,
\end{equation}
where
\(
\bepsi \sim \mathcal{N}(\bzero_{2},\bI_2)\), \(
\bet \sim \mathcal{N}(\bzero_{2},\bI_2-\bGamma\bGamma'),
\). By construction
\(
\blm = \bz_p - \bGamma \bz_a
\),
\(
\blm = \bh + \bdelta_m + \bet\), and \(
\bdelta_m = \bmu_p - \bGamma \bmu_a.
\)

Our proposed original statistic is given by (in the limit experiment)
\begin{equation}
R_{\mathrm{orig}}
=
\frac{\bz_p'\bz_a}{\sqrt{\operatorname{tr}(\bGamma^2)+\|\bz_p\|^2}},
\end{equation}
while the pivotal (limiting) score statistic as in \citet{kleibergen2002pivotal} is of the form
\begin{equation}
R_K = \frac{\blm'\bz_a}{\|\blm\|}.
\end{equation}

We consider two alternatives. The first is designed to favor the original statistic where we set \(
\bmu_p = \bGamma \bmu_a.
\)
In this case,
\(
\bdelta_m = \bmu_p - \bGamma \bmu_a = \bzero,
\)
so
\(
\blm = \bh+\bet.
\)
When \(\bh=\bzero\), the score statistic does not contain a deterministic alignment between \(\blm\) and \(\bz_a\). By contrast, since
\(
\bz_p = \blm+\bGamma \bz_a,
\)
the original numerator can be decomposed as
\(
\bz_p'\bz_a = \blm'\bz_a + \bz_a'\bGamma \bz_a = \blm'\bz_a + c\|\bz_a\|^2.
\)
The original statistic therefore contains the same residualized-correlation component as \(R_K\), plus a positive quadratic term in \(\bz_a\). 

The second alternative is designed to favor the score statistic. We set
\(
\bmu_p = \bzero.
\)
Then
\(
\bdelta_m = -\bGamma \bmu_a = -c\bmu_a.
\)
Thus, when \(\bh=\bzero\),
\(
\E[\bz_{p}'\bz_a] = \bdelta_m'\bmu_a = -c\|\bmu_a\|^2.
\)
The score statistic therefore receives a direct deterministic signal from the alignment between \(\blm\) and \(\bz_a\). By contrast, the original statistic does not have a mean-product signal of the form \(\bmu_p'\bmu_a\), since \(\bmu_p=\bzero\). 

For the local power analysis, we parameterize \(
\bmu_a = 2\theta \bv,\)
where \(\theta\) indexes the strength of the local alternative, and the fixed direction $\bv$ is normalized so that \(\|\bv\|=1\). Figure~\ref{figapp:power} shows that indeed neither test dominates the other.

\begin{figure}[t]
\centering
\caption{Power comparison in the isotropic experiment.}
\label{figapp:power}
\includegraphics[width=0.48\textwidth]{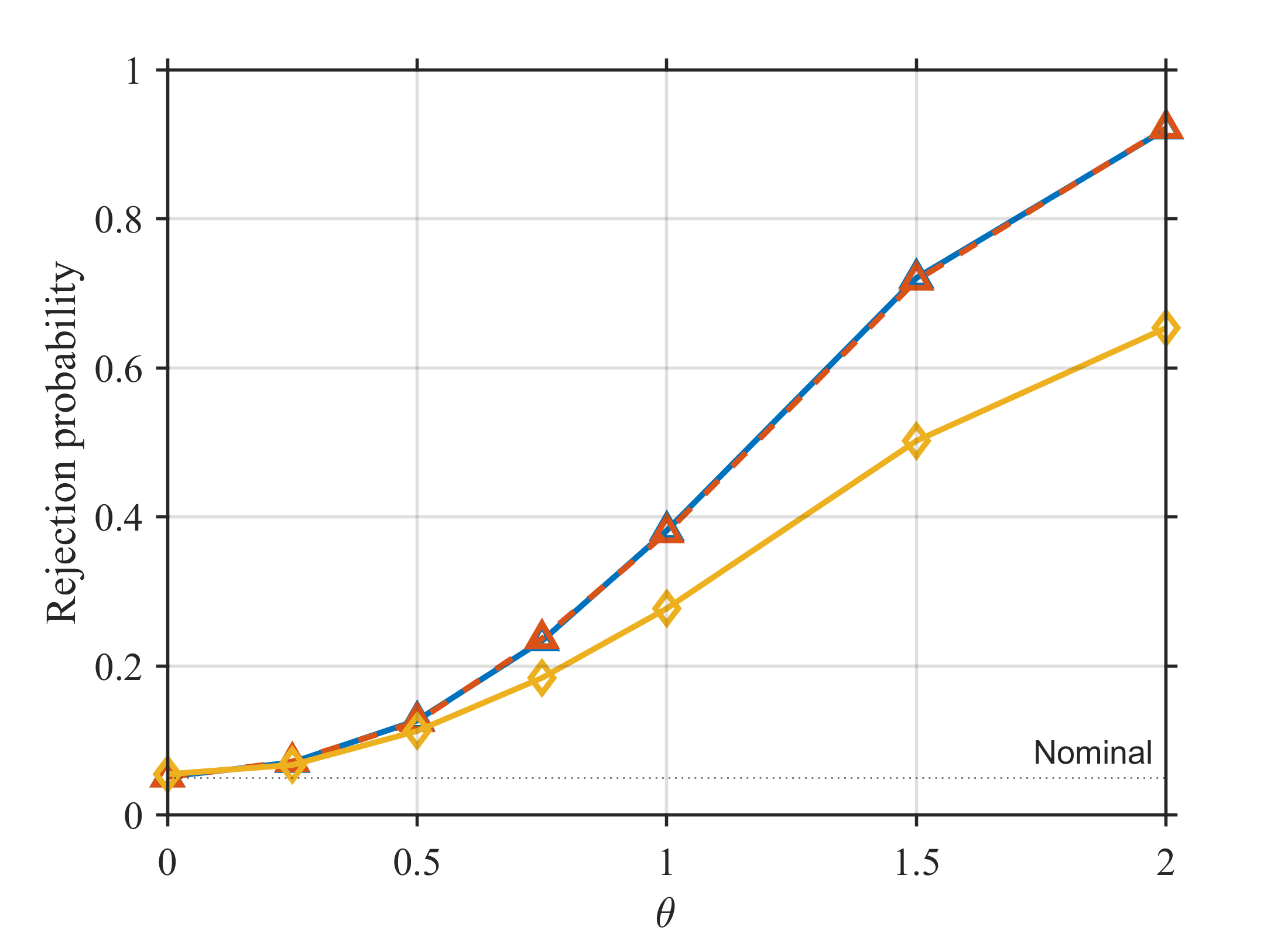}
\includegraphics[width=0.48\textwidth]{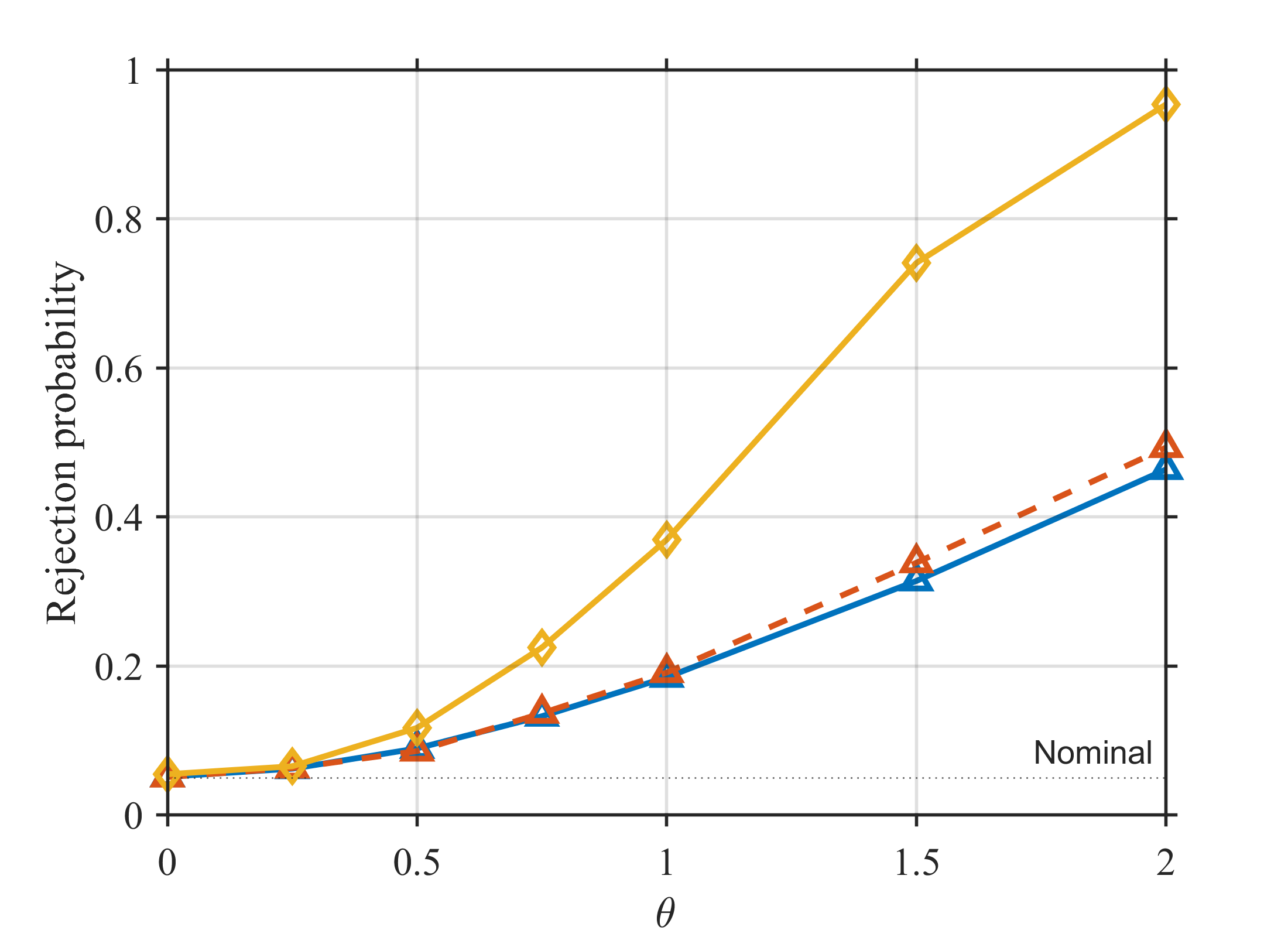}
\begin{minipage}{0.95\textwidth}
\footnotesize\textit{Note}: On the left, the design where the score test (yellow line) has lower power relative to the conditional test (red). On the right, the design where the score test has higher power relative to the conditional test. The red dashed line uses conditional critical values, the blue solid line uses standard normal critical values.
\end{minipage}
\end{figure}

\end{document}

%% file: packages_and_commands.tex
\usepackage{amsmath}
\usepackage{amssymb}
\usepackage{amsthm}
\usepackage[authoryear]{natbib}
\usepackage{booktabs}
\usepackage{graphicx}
\usepackage{epstopdf}
\usepackage{caption}
\usepackage{bm}
\usepackage{bbm}
\usepackage{xcolor}
\usepackage{pdflscape}
\usepackage{longtable}
\usepackage{threeparttablex}
\usepackage{mathrsfs}
\usepackage[english]{babel}
\usepackage{subcaption}
\usepackage[a4paper,margin=1in]{geometry}
\usepackage[title]{appendix}
\usepackage{float}
\usepackage[doublespacing]{setspace}
\usepackage[]{threeparttable}
\usepackage{multirow}
\usepackage{embolds}

\usepackage[hypertexnames=false]{hyperref}
\definecolor{lightblue}{rgb}{0, 0.4, 0.6}
\hypersetup{
    breaklinks=true,
	colorlinks = true,
	citecolor = {lightblue},
	linkcolor = {lightblue},
	urlcolor = {lightblue}
}

\usepackage[capitalise,nameinlink,noabbrev]{cleveref}
\crefname{equation}{}{}
\crefname{appsec}{Appendix}{Appendices}
\crefname{appsubsec}{Appendix}{Appendices}

\usepackage[inline,shortlabels]{enumitem}
\newlist{assenumerate*}{enumerate*}{4} 
\setlist[assenumerate*]{label=(\roman*), ref=\theassumption.\roman*}
\crefname{assenumerate*i}{Assumption}{Assumptions}

\newlist{assenumerate}{enumerate}{4} 
\setlist[assenumerate]{label=(\roman*), ref=\theassumption.\roman*}
\crefname{assenumeratei}{Assumption}{Assumptions}

\usepackage{tikz}
\usetikzlibrary{shapes}
\usetikzlibrary{decorations.pathreplacing}

\newtheorem{counter}{counter} 
\newtheorem{theorem}[counter]{Theorem}
\crefname{theorem}{Theorem}{Theorems}

\crefname{conjecture}{Conjecture}{Conjectures}

\crefname{lemma}{Lemma}{Lemmas}

\newtheorem{remark}[counter]{Remark}
\crefname{definition}{Definition}{Definitions}
\newtheorem{assumption}{Assumption}
\renewcommand{\theassumption}{A\arabic{assumption}} 
\crefname{assumption}{Assumption}{Assumptions}
\newtheorem{corollary}[counter]{Corollary}
\crefname{corollary}{Corollary}{Corollaries}
\newtheorem{proposition}[counter]{Proposition}
\crefname{proposition}{Proposition}{Propositions}
\theoremstyle{definition}

\DeclareMathOperator{\E}{E}

\newtheoremstyle{compactthm}
{0.6ex plus 0.2ex minus 0.1ex} 
{0.6ex plus 0.2ex minus 0.1ex} 
{\itshape}                     
{0pt}                          
{\bfseries}                    
{.}                            
{0.5em}                        
{}                             

\theoremstyle{compactthm}

\usepackage{titlesec}

\setlist{topsep=0.3ex, itemsep=0.2ex, parsep=0pt, partopsep=0pt}

%% file: bibliography.bib
@article{jung2026testing,
  title={Testing for peer effects without specifying the network structure},
  author={Jung, Hyunseok and Liu, Xiaodong},
  journal={Journal of Econometrics},
  volume={253},
  pages={106124},
  year={2026},
  publisher={Elsevier}
}

@article{mikusheva2022inference,
  title={Inference with many weak instruments},
  author={Mikusheva, Anna and Sun, Liyang},
  journal={Review of Economic Studies},
  volume={89},
  number={5},
  pages={2663--2686},
  year={2022},
  publisher={Oxford University Press}
}

@article{bekker1994alternative,
  title={Alternative approximations to the distributions of instrumental variable estimators},
  author={Bekker, Paul A},
  journal={Econometrica},
  volume={62},
  number={3},
  pages={657--681},
  year={1994},
  publisher={JSTOR}
}

@article{ANTOINE2012350,
title = {Efficient minimum distance estimation with multiple rates of convergence},
journal = {Journal of Econometrics},
volume = {170},
number = {2},
pages = {350-367},
year = {2012},
note = {Thirtieth Anniversary of Generalized Method of Moments},
issn = {0304-4076},
doi = {https://doi.org/10.1016/j.jeconom.2012.05.010},
url = {https://www.sciencedirect.com/science/article/pii/S0304407612001224},
author = {Bertille Antoine and Eric Renault}
}

@article{staigerstock1997,
 ISSN = {00129682, 14680262},
 URL = {http://www.jstor.org/stable/2171753},
 author = {Douglas Staiger and James H. Stock},
 journal = {Econometrica},
 number = {3},
 pages = {557--586},
 publisher = {[Wiley, Econometric Society]},
 title = {Instrumental Variables Regression with Weak Instruments},
 urldate = {2026-09-04},
 volume = {65},
 year = {1997}
}

@article{Juodis2022CCER,
author = {Juodis, Artūras},
doi = {10.1002/jae.2899},
issn = {10991255},
journal = {Journal of Applied Econometrics},
	volume={37},
	number={4},
pages = {788--810},
title = {{A regularization approach to common correlated effects estimation}},
year = {2022}
}

@article{anderson1949estimation,
  title={Estimation of the parameters of a single equation in a complete system of stochastic equations},
  author={Anderson, Theodore W and Rubin, Herman},
  journal={Annals of Mathematical Statistics},
  volume={20},
  number={1},
  pages={46--63},
  year={1949},
  publisher={Institute of Mathematical Statistics}
}

@article{kleibergen2002pivotal,
  title={Pivotal statistics for testing structural parameters in instrumental variables regression},
  author={Kleibergen, Frank},
  journal={Econometrica},
  volume={70},
  number={5},
  pages={1781--1803},
  year={2002},
  publisher={Wiley Online Library}
}

@article{moreira2003conditional,
  title={A conditional likelihood ratio test for structural models},
  author={Moreira, Marcelo J},
  journal={Econometrica},
  volume={71},
  number={4},
  pages={1027--1048},
  year={2003},
  publisher={Wiley Online Library}
}

@article{pesaran2015testing,
  title={Testing weak cross-sectional dependence in large panels},
  author={Pesaran, M Hashem},
  journal={Econometric Reviews},
  volume={34},
  number={6-10},
  pages={1089--1117},
  year={2015},
  publisher={Taylor \& Francis}
}

@article{pesaran2021general,
  title={General diagnostic tests for cross-sectional dependence in panels},
  author={Pesaran, M Hashem},
  journal={Empirical Economics},
  volume={60},
  number={1},
  pages={13--50},
  year={2021},
  publisher={Springer}
}

@article{bramoulle2009identification,
  author  = {Bramoull{\'e}, Yann and Djebbari, Habiba and Fortin, Bernard},
  title   = {Identification of Peer Effects through Social Networks},
  journal = {Journal of Econometrics},
  year    = {2009},
  volume  = {150},
  number  = {1},
  pages   = {41--55}
}

@article{goldsmithpinkham2013social,
  author  = {Goldsmith-Pinkham, Paul and Imbens, Guido W.},
  title   = {Social Networks and the Identification of Peer Effects},
  journal = {Journal of Business \& Economic Statistics},
  year    = {2013},
  volume  = {31},
  number  = {3},
  pages   = {253--264}
}

@article{bramoulle2020peer,
  author  = {Bramoull{\'e}, Yann and Djebbari, Habiba and Fortin, Bernard},
  title   = {Peer Effects in Networks: A Survey},
  journal = {Annual Review of Economics},
  year    = {2020},
  volume  = {12},
  pages   = {603--629}
}

@unpublished{manresa2016estimating,
  author = {Manresa, Elena},
  title  = {Estimating the Structure of Social Interactions Using Panel Data},
  note   = {Working paper},
  year   = {2016}
}

@article{depaula2025identifying,
  title={Identifying network ties from panel data: Theory and an application to tax competition},
  author={De Paula, Aureo and Rasul, Imran and Souza, Pedro CL},
  journal={Review of Economic Studies},
  volume={92},
  number={4},
  pages={2691--2729},
  year={2025},
  publisher={Oxford University Press UK}
}

@unpublished{yauck2022estimation,
  author = {Yauck, Mamadou},
  title  = {On the Estimation of Peer Effects for Sampled Networks},
  note   = {Working paper},
  year   = {2022}
}

@article{breusch1980lagrange,
  author  = {Breusch, Trevor S. and Pagan, Adrian R.},
  title   = {The {L}agrange {M}ultiplier Test and its Applications to Model Specification in Econometrics},
  journal = {Review of Economic Studies},
  year    = {1980},
  volume  = {47},
  number  = {1},
  pages   = {239--253}
}

@article{phillips1999linear,
  title={Linear regression limit theory for nonstationary panel data},
  author={Phillips, Peter CB and Moon, Hyungsik R},
  journal={Econometrica},
  volume={67},
  number={5},
  pages={1057--1111},
  year={1999},
  publisher={Wiley Online Library}
}

@article{Hsiao02102022,
author = {Cheng Hsiao and Zhentao Shi and Qiankun Zhou},
title = {Transformed Estimation for Panel Interactive Effects Models},
journal = {Journal of Business \& Economic Statistics},
volume = {40},
number = {4},
pages = {1831--1848},
year = {2022},
publisher = {Taylor \& Francis},
doi = {10.1080/07350015.2021.1983438}
}

@BOOK{magnus,
  title = {Matrix Differential Calculus with Applications in Statistics and
	Econometrics},
  publisher = {John Wiley \& Sons},
  year = {2007},
  author = {J. R. Magnus and H. Neudecker}
}

@article{yin2021focused,
  title={Focused information criterion and model averaging for large panels with a multifactor error structure},
  author={Yin, Shou-Yung and Liu, Chu-An and Lin, Chang-Ching},
  journal={Journal of Business \& Economic Statistics},
  volume={39},
  number={1},
  pages={54--68},
  year={2021},
  publisher={Taylor \& Francis}
}

@article{voigtlander2014skill,
  title={Skill bias magnified: Intersectoral linkages and white-collar labor demand in US manufacturing},
  author={Voigtl{\"a}nder, Nico},
  journal={Review of Economics and Statistics},
  volume={96},
  number={3},
  pages={495--513},
  year={2014},
  publisher={The MIT Press}
}

@article{juodis2022incidental,
  author  = {Juodis, Art{\=u}ras and Reese, Simon},
  title   = {The Incidental Parameters Problem in Testing for Remaining Cross-Section Correlation},
  journal = {Journal of Business \& Economic Statistics},
  year    = {2022},
  volume  = {40},
  number  = {3},
  pages   = {1191--1203}
}

@article{hounyo2025serial,
  title={Serial dependence robust bootstrap test for cross-sectional correlation},
  author={Hounyo, Ulrich and Kao, Chihwa and Kim, Min Seong},
  journal={Econometrics Journal},
  volume={28},
  number={3},
  pages={442--461},
  year={2025},
  publisher={Oxford University Press}
}

@article{PesaranXie2026, 
    title={HOW TO DETECT NETWORK DEPENDENCE IN LATENT FACTOR MODELS? {A} BIAS-CORRECTED {CD} TEST}, 
    DOI={10.1017/S0266466625100212}, 
    journal={Econometric Theory}, 
    author={Pesaran, M. Hashem and Xie, Yimeng}, 
    year={2026}, pages={1–41}}
